\documentclass[11pt]{article}
\usepackage[utf8]{inputenc}
\usepackage{latexsym,amsfonts,amsmath,amsthm,amssymb, epsfig}
\usepackage[a4paper,total={170mm,250mm},vcentering,dvips]{geometry}
\usepackage{hyperref,color}
\usepackage{wrapfig}

\usepackage{nicematrix}
\usepackage{arydshln,nccrules}

\newtheorem{thm}{Theorem}[section]

\newtheorem{lemma}[thm]{Lemma}

\newtheorem{dfn}[thm]{Definition}
\newtheorem{proposition}[thm]{Proposition}

\newtheorem{rem}[thm]{Remark}

\DeclareMathOperator{\ReLU}{ReLU}
\DeclareMathOperator{\sign}{sign}
\DeclareMathOperator{\supp}{supp}
\newcommand{\CalC}{\mathcal C}
\newcommand{\CalS}{\mathcal S}

\newcommand{\R}{\mathbb{R}}

\newcommand{\N}{\mathbb{N}}

\newcommand{\Z}{\mathbb{Z}}

\usepackage{trimclip}
\newlength{\trianglewidth}
\newlength{\pluswidth}
\newcommand{\righttriangleplus}{%
    \mathrel{\makebox[\trianglewidth]{%
        \raisebox{.3\height}{\clipbox{.3\width{} .3\height}{\(+\)}}%
        \hspace*{-.333\pluswidth}%
        \makebox[\trianglewidth]{\(>\)}%
    }}%
}

\begin{document}

\title{A multivariate Riesz basis of ReLU neural networks}
\author{Cornelia Schneider\footnote{Friedrich-Alexander Universit\"at Erlangen, Applied Mathematics III, Cauerstr. 11, 91058 Erlangen, Germany. Email: \href{mailto:cornelia.schneider@math.fau.de}{cornelia.schneider@math.fau.de}}
\ 
and Jan Vyb\'iral\footnote{Department of Mathematics, Faculty of Nuclear Sciences and Physical Engineering,
Czech Technical University, Trojanova 13, 12000 Praha, Czech Republic. Email: \href{mailto:jan.vybiral@fjfi.cvut.cz}{jan.vybiral@fjfi.cvut.cz}.
The work of this author has been supported by the grant P202/23/04720S of the Grant Agency of the Czech Republic}}
\maketitle
\begin{abstract}
We consider the trigonometric-like system of piecewise linear functions introduced recently by Daubechies, DeVore, Foucart, Hanin, and Petrova.
We provide an alternative proof that this system forms a Riesz basis of $L_2([0,1])$ based on the Gershgorin theorem. We also generalize
this system to higher dimensions $d>1$ by a construction, which avoids using (tensor) products. As a consequence, the functions from the new Riesz basis
of $L_2([0,1]^d)$ can be easily represented by neural networks. Moreover, the Riesz constants of this system are independent of $d$, making it an attractive
building block regarding  future multivariate analysis of neural networks.
{\em MSC2020 Math Subject Classifications:} 68T07, 42C15, 11A25.
\end{abstract}

\section{Introduction}

The last decades observed a tremendous success of artificial neural networks in many machine learning
tasks, including computer vision \cite{Intro3}, speech recognition \cite{Intro4},
natural language processing \cite{IntroA}, or games solutions \cite{Intro7, Intro6} to name just few.
Despite their wide use, many of their properties are not fully understood
and many aspects of their great practical performance lack a rigorous explanation. Without any doubt,
a deeper insight into the theory of artificial neural networks could boost their applicability even further.

In the last decade, a growing number of authors investigated why the deep neural networks with a higher number
of hidden layers approximate many interesting functions more efficiently than the shallow neural networks
(with only one hidden layer) using the same number of parameters. We refer to \cite{BGKP,DDFHP, DeVore,EPGB,GKNV,PV,Telg,Yarot} 
for a number of mathematically rigorous results in this direction.

One of the astonishing properties of artificial neural networks is, that they can approximate
extremely well also functions of many variables, which often allows to avoid the curse of dimensionality \cite{Jentzen,DN,Poggio}.
The aim of this work is to shed new light on the effectiveness of artificial neural networks
for the approximation of multivariate functions by constructing a new system of functions,
which forms a Riesz basis of $L_2([0,1]^d)$ for every $d\ge 1$.

To state the result, we first recall the notion of a Riesz basis (which, in turn, 
is a generalization of an orthonormal system and of an orthonormal basis, cf. \cite{Ole}).
\begin{dfn}\label{dfn:Riesz}
Let $H$ be a real Hilbert space. The (finite or infinite) sequence $(x_n)_n\subset H$ is called a Riesz sequence if there are two constants $A,B>0$ such that
\begin{equation}\label{eq:Riesz}
A\sum_n\alpha_n^2\le\left\|\sum_n \alpha_nx_n\right\|^2\le B\sum_n\alpha_n^2
\end{equation}
for every real square summable sequence $(\alpha_n)_n$. If the closed span of $(x_n)_n$ is the whole space $H$, then we call it a Riesz basis.
\end{dfn}

The system, which we study in this paper, is a trigonometric-like basis, where instead of $\cos$ and $\sin$ functions
we use their piecewise linear counterparts $\CalC$ and $\CalS$, which are defined as follows (cf. Figure \ref{fig:1}).
\begin{dfn}\begin{enumerate}
\item For $x\in[0,1]$, we define
\[
{\CalC} (x)=4\left|x-\frac{1}{2}\right|-1
=\begin{cases}1-4x,\ x\in[0,1/2),\\4x-3,\ x\in[1/2,1]\end{cases}
\]
and
\[
{\CalS}(x)=\left|2-4\left|x-\frac{1}{4}\right|\right|-1=\begin{cases}
4x,\ &x\in[0,1/4),\\
2-4x,\ &x\in[1/4,3/4),\\
4x-4, &x\in [3/4,1].
\end{cases}
\]
\item For $x\in\R$, we extend this definition periodically, i.e. $\CalC(x)=\CalC(x-\lfloor x\rfloor)$ and $\CalS(x)=\CalS(x-\lfloor x\rfloor).$
\item If $k\ge 1$ and $x\in\R$, we put $\CalC_k(x)=\CalC(kx)$ and $\CalS_k(x)=\CalS(kx).$
\end{enumerate}
\end{dfn}

\begin{figure}[h]
\begin{minipage}{0.45\textwidth}
\includegraphics[width=9cm]{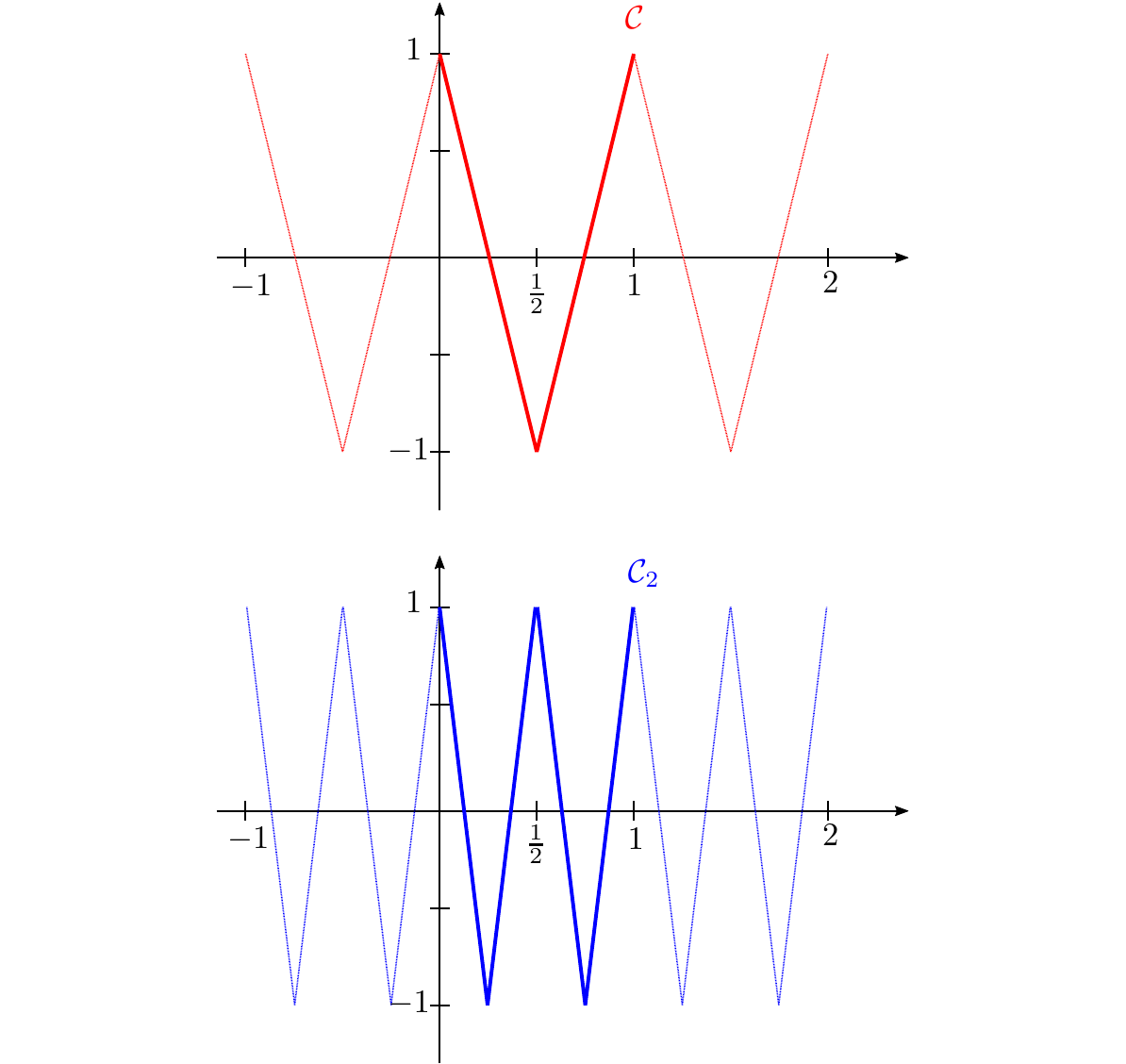}
\end{minipage}\hfill 
\begin{minipage}{0.45\textwidth}
\includegraphics[width=9cm]{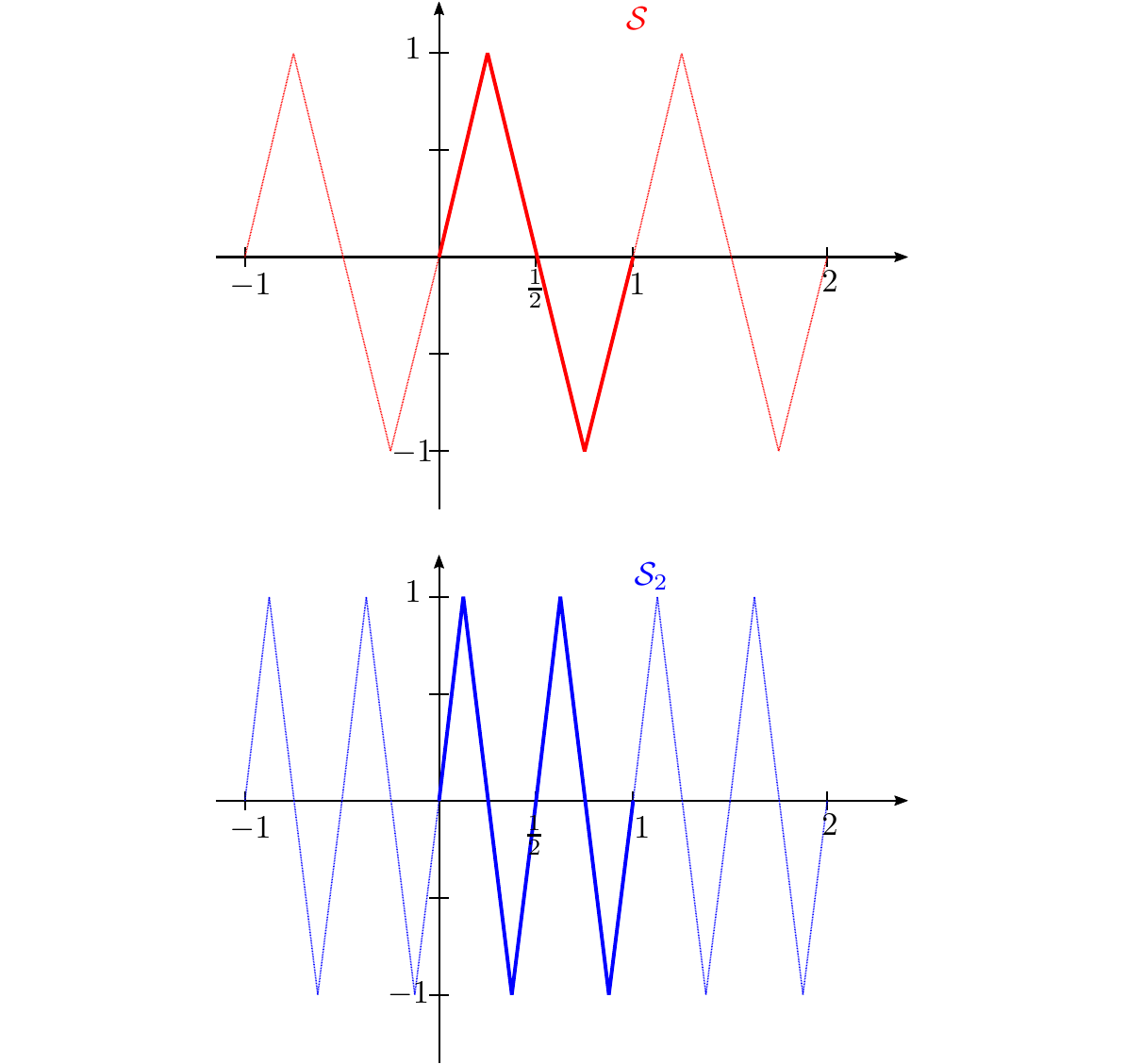}
\end{minipage}
\caption{The plot of $\CalC$, $\CalS$, $\CalC_2$ and $\CalS_2$.}
\label{fig:1}
\end{figure}

These functions were introduced and studied in \cite{DDFHP}, where it was shown that the system
\begin{equation}\label{eq:R1}
{\overline{\mathcal R}}_1:=\{1\}\cup\{\sqrt{3}\,\CalC_k,\sqrt{3}\,\CalS_k:k\in\N\}
\end{equation}
forms a Riesz basis of $L_2([0,1])$ with the constants $A=1/2$ and $B=3/2.$
Let us note that the factor $\sqrt{3}$ 
in \eqref{eq:R1} is simply a normalization factor, which ensures that all the elements of ${\overline{\mathcal R}}_1$
have unit norm in $L_2([0,1])$. The proof given in \cite{DDFHP} is implicitly inspired by  the method
of analysis  and synthesis operators, respectively,  used in the frame theory, cf. \cite{Charlie}.
It is one of the aims of our paper (cf. Theorem \ref{thm:CSd1}) to provide an alternative proof, which first reduces \eqref{eq:Riesz}
to the study of spectral properties of the Gram matrix of ${\overline{\mathcal R}}_1$. The result then follows from
the Gershgorin circle theorem and some elementary number theory (including Euler products and a certain Ramanujan's formula).

The main advantage of \eqref{eq:R1} in contrast to the standard trigonometric system is, that its elements
can be easily identified by artificial neural networks with the REctified Linear Unit (ReLU) activation function.\\

\begin{wrapfigure}{r}{0.35\textwidth}
\includegraphics[width=5cm]{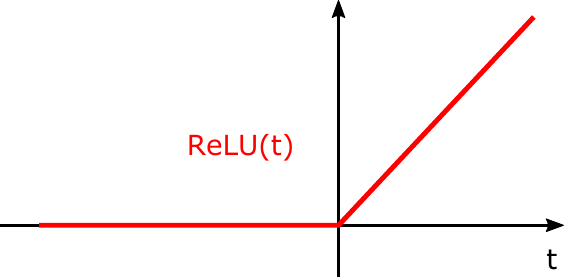}
\caption{{\small Graph of the $\ReLU$ function.}}
\label{fig:ReLU}
\end{wrapfigure}
Let us recall, that if $t\in\R$, then the ReLU function is defined as $\ReLU(t)=\max(0,t)$, cf. Figure \ref{fig:ReLU}. On vectors, it acts component-wise
\[
\ReLU(x_1,\dots,x_n)=(\ReLU(x_1),\dots,\ReLU(x_n)),\quad x\in \R^n.
\]

If $W\ge 2$ and $L\ge 1$ are integer parameters, then we denote by $\Upsilon^{W,L}$ the real-valued functions
 which can be represented by a $\ReLU$ neural network of width $W$ and depth $L$,
see Definition \ref{dfn:NN} for a precise formulation. Then \cite[Theorem 6.2]{DDFHP}
shows that $\CalC_j$ and $\CalS_j$ (restricted to $[0,1]$) lie in $\Upsilon^{W,L}$ for an arbitrary $W\ge 6$
and $L$ of the asymptotic order $\log_2(j).$

The main aim of our work is to generalize the results of \cite{DDFHP} to the multivariate case.
There are two crucial issues which prevent us from simply taking the tensor products of the functions in 
${\overline{\mathcal R}}_1$. First, the product function $(x,y)\to x\cdot y$ can only be approximated by the $\ReLU$ neural networks
and, second, the ratio of the Riesz constants $B$ and $A$ gets exponentially large when $d$ grows.

We propose a surprisingly simple and effective solution to these challenges.
We show that (cf. Theorem \ref{thm:CSd}) the multivariate analogue of ${\overline{\mathcal R}}_1$
\begin{equation}\label{eq:Rd}
\{1\}\cup\{\sqrt{3}\,{\mathcal C}(\alpha\cdot x):\alpha\righttriangleplus 0\}
\cup\{\sqrt{3}\,\CalS(\alpha\cdot x):\alpha\righttriangleplus 0\},
\end{equation}
forms a Riesz basis of $L_2([0,1]^d)$ for every $d\ge 1$ with the same constants $A=1/2$ and $B=3/2$
as in the univariate case. Here, $\alpha\righttriangleplus 0$
means that the first non-zero entry of $\alpha=(\alpha_1,\dots,\alpha_d)\in\Z^d$ is positive.
Finally, we show in Section \ref{sec:NN} that the functions from \eqref{eq:Rd} can be exactly reproduced
by the $\ReLU$ neural networks of width $W$ and depth $L$ in essentially the same way as in the univariate case,
i.e., there is virtually no price to pay when $d$ grows.

\section{Univariate case}

It was observed in \cite[Section 6]{DDFHP}, that the system of piecewise linear functions
\begin{equation}\label{eq:Daub1}
\{\CalC_k,\CalS_k:k\in\N\}, 
\end{equation}
 on one hand shares some nice properties with the trigonometric system and on the other hand  can be easily reproduced by artificial neural networks with the $\ReLU$ activation function.
The aim of this section is to essentially reprove Proposition 6.1 of \cite{DDFHP}, which states that this system
is a Riesz basis of $L_2^0([0,1])$, the space of square integrable functions with mean zero.

Although our proof shares some technical details with \cite{DDFHP}, its main structure is different: In particular, 
we reduce the problem to spectral properties of the corresponding Gram matrix and then apply the Gershgorin circle theorem.
Interestingly, by using some elementary number theory, we are able to  completely characterize the inner products of the $\CalC_i$  and/or  $\CalS_j$  functions.
In contrast to \cite{DDFHP}, we complement \eqref{eq:Daub1} by adding the constant function,
which is orthogonal to all functions from \eqref{eq:Daub1}.

We denote by $\gcd(i,j)$ the greatest common divisor of $i$ and $j$ and by $\langle f,g\rangle=\int_0^1 f(t)g(t)dt$
the standard inner product in $L_2([0,1]).$

\begin{lemma}\label{lem:inner_productC} Let $i,j\in\N.$ Then
\begin{enumerate}
\item $\langle \CalC_i,\CalS_j\rangle=0$;
\item $\langle {\CalC}_i,{\CalC}_j\rangle=\langle\CalS_i,\CalS_j\rangle=0$ if $i/\gcd(i,j)$ is odd and $j/\gcd(i,j)$ is even (or vice versa), i.e.,
if the prime factorizations of $i$ and $j$ contain a different power of 2;
\item If $i/\gcd(i,j)$ and $j/\gcd(i,j)$ are both odd, then
\[
3\cdot \langle {\CalC}_i,{\CalC}_j\rangle=3\cdot |\langle {\CalS}_i,{\CalS}_j\rangle|
=\frac{\gcd(i,j)^4}{i^2\cdot j^2}.
\]
Here, the sign of $\langle {\CalS}_i,{\CalS}_j\rangle$ is negative if, and only if, $(i+j)/(2\gcd(i,j))$ is even.
\item In particular, we get $\langle {\CalC}_i,{\CalC}_i\rangle=\langle \CalS_i,\CalS_i\rangle=1/3$ for all $i\in \N$.
\end{enumerate}
\end{lemma}

\begin{proof}
We transfer the proof to the Fourier side by exploiting the decomposition of $\CalC_i$ and $\CalS_j$ into Fourier series, cf. \cite[page 166]{DDFHP}. Let
\[
c_k(x)=\sqrt{2}\cos(2\pi kx),\quad s_k(x)=\sqrt{2}\sin(2\pi kx),\quad k\in\N_0,\ x\in\R.
\]
Then a standard calculation reveals that
\begin{equation}\label{eq:CSFourier}
\sqrt{3}\,{\mathcal C}_k=\mu\sum_{m\ge 0}\frac{1}{(2m+1)^2}c_{(2m+1)k}\quad\text{and}\quad 
\sqrt{3}\,{\mathcal S}_k=\mu\sum_{m\ge 0}\frac{(-1)^m}{(2m+1)^2}s_{(2m+1)k},
\end{equation}
where
\begin{equation}\label{eq:mu}
\mu^2\sum_{m\ge 0}\frac{1}{(2m+1)^4}=1, \quad \text{i.e.}\quad \mu^2\frac{\pi^4}{96}=1.
\end{equation}
Using \eqref{eq:CSFourier}, we immediately obtain that $\langle\CalC_i,\CalS_j\rangle=0$.
Furthermore,
\begin{align}
3\langle {{\mathcal C}_i},{{\mathcal C}_j}\rangle&=\sum_{m,n=0}^\infty \frac{\mu^2}{(2m+1)^2(2n+1)^2}\langle c_{(2m+1)i},c_{(2n+1)j}\rangle \notag \\
&=\sum_{m,n=0}^\infty \frac{\mu^2}{(2m+1)^2(2n+1)^2}\delta_{(2m+1)i,(2n+1)j}, \label{prodC_iC_j}
\end{align}
where $\delta_{u,v}=1$ if $u=v$ and zero otherwise.
To simplify \eqref{prodC_iC_j}, we have to find for fixed $i,j\in\N$ all $m,n\in\N_0$ such that $(2m+1)i=(2n+1)j$.
First, we observe that if the prime factorizations of $i$ and $j$ contain a different power of two,
then also $(2m+1)i$ and $(2n+1)j$ have a different power of two in their prime factorizations and therefore they differ for all $m,n\in\N_0$. Consequently, \eqref{prodC_iC_j} shows that $3\langle {{\mathcal C}_i},{{\mathcal C}_j}\rangle=0$.

If the prime factorizations of $i$ and $j$ contain the same power of two, then $i/\gcd(i,j)$ and $j/\gcd(i,j)$ are both odd.
We denote $g=\gcd(i,j)$ and note that $i/g$ and $j/g$ are coprime, i.e., that their greatest common divisor is one. 
We then look for all pairs $(m,n)\in\N^2_0$, which solve the equation
\[
(2m+1)\cdot g\cdot\frac{i}{g}=(2n+1)\cdot g\cdot\frac{j}{g}.
\]
All the solutions are obtained in the form
\begin{equation}\label{eq:ml}
2m+1=\frac{j}{g}\cdot(2l+1),\quad 2n+1=\frac{i}{g}\cdot (2l+1),\quad l\in\N_0.
\end{equation}
We insert \eqref{eq:ml} into \eqref{prodC_iC_j} and conclude that 
\begin{align*}
3\langle {{\CalC}_i},{{\CalC}_j}\rangle&=\sum_{l=0}^\infty \frac{\mu^2}{\displaystyle \left(\frac{i}{g}(2l+1)\right)^2\cdot \left(\frac{j}{g}(2l+1)\right)^2}\\
&=\mu^2\sum_{l=0}^\infty \frac{1}{(2l+1)^4}\cdot \frac{1}{(i/g)^2\cdot (j/g)^2}= \frac{1}{(i/g)^2\cdot (j/g)^2}.
\end{align*}

The calculation of $\langle\CalS_i,\CalS_j\rangle$ can be performed in a very similar way, one only needs to take care about the sign of the inner product. In particular, instead of \eqref{prodC_iC_j} we obtain 
\[
3\langle {{\mathcal S}_i},{{\mathcal S}_j}\rangle
=\sum_{m,n=0}^\infty \frac{\mu^2(-1)^{n+m}}{(2m+1)^2(2n+1)^2}\delta_{(2m+1)i,(2n+1)j}, 
\]
where $(-1)^{n+m}=-1$ if, and only if, $n+m=\frac{i+j}{2g}(2l+1)-1$ is odd. Therefore, the sign is negative if, and only if,  $\frac{i+j}{2g}$ even. 
\end{proof}

Next, we combine Lemma \ref{lem:inner_productC} with the Gershgorin circle theorem
and provide an alternative proof of \cite[Theorem 6.2]{DDFHP},
which we restate as follows.

\begin{thm}\label{thm:CSd1}
The system ${\mathcal R}_1:=\{1\}\cup\{\CalC_k,\CalS_k:k\in\N\}$ is a Riesz basis of $L_2([0,1])$.
\end{thm}

Before we come to the proof, several remarks seem to be in order.

\begin{rem}
 \begin{enumerate}
 \item We will actually show, that the Riesz constants of the $L_2$-normalized system
 ${\overline{\mathcal R}}_1:=\{1\}\cup\{\sqrt{3}\,\CalC_k,\sqrt{3}\,\CalS_k:k\in\N\}$
 can be chosen as $A=1/2$ and $B=3/2.$
 \item We divide the proof of Theorem \ref{thm:CSd1} into several steps. In the first two steps, we show that the truncated
 system ${\overline{\mathcal R}}_1^N:=\{1\}\cup \{\sqrt{3}\,\CalC_k,\sqrt{3}\,\CalS_k:k\le N\}$ forms a Riesz sequence
 (i.e., that it satisfies \eqref{eq:Riesz}) with $A=1/2$ and $B=3/2$  being independent  of  $N\in\N$. The first step
 reduces this question to spectral properties of the Gram matrix of ${\overline{\mathcal R}}_1^N$ and in the second step
 we apply the Gershgorin theorem to bound this spectrum. The third step describes how we pass to the limit $N\to\infty$
 to deduce that ${\overline{\mathcal R}}_1$ is also a Riesz sequence. Finally, the fourth step shows that ${\overline{\mathcal R}}_1$
 is also a basis, i.e., that its closed linear span is $L_2([0,1])$.
\item Our proof shows that the spectrum of the Gram matrix of ${\overline{\mathcal R}}_1^N$ (for arbitrary $N\in \N$) is contained in $[1/2,3/2].$
We leave it as an open problem to find out if these bounds are actually optimal. Supported by   numerical evidence
(see Figure \ref{fig:Gram} for details), our conjecture is that there is indeed some space for improvement.
\begin{figure}[h!]
\begin{center}\includegraphics[width=6.5cm]{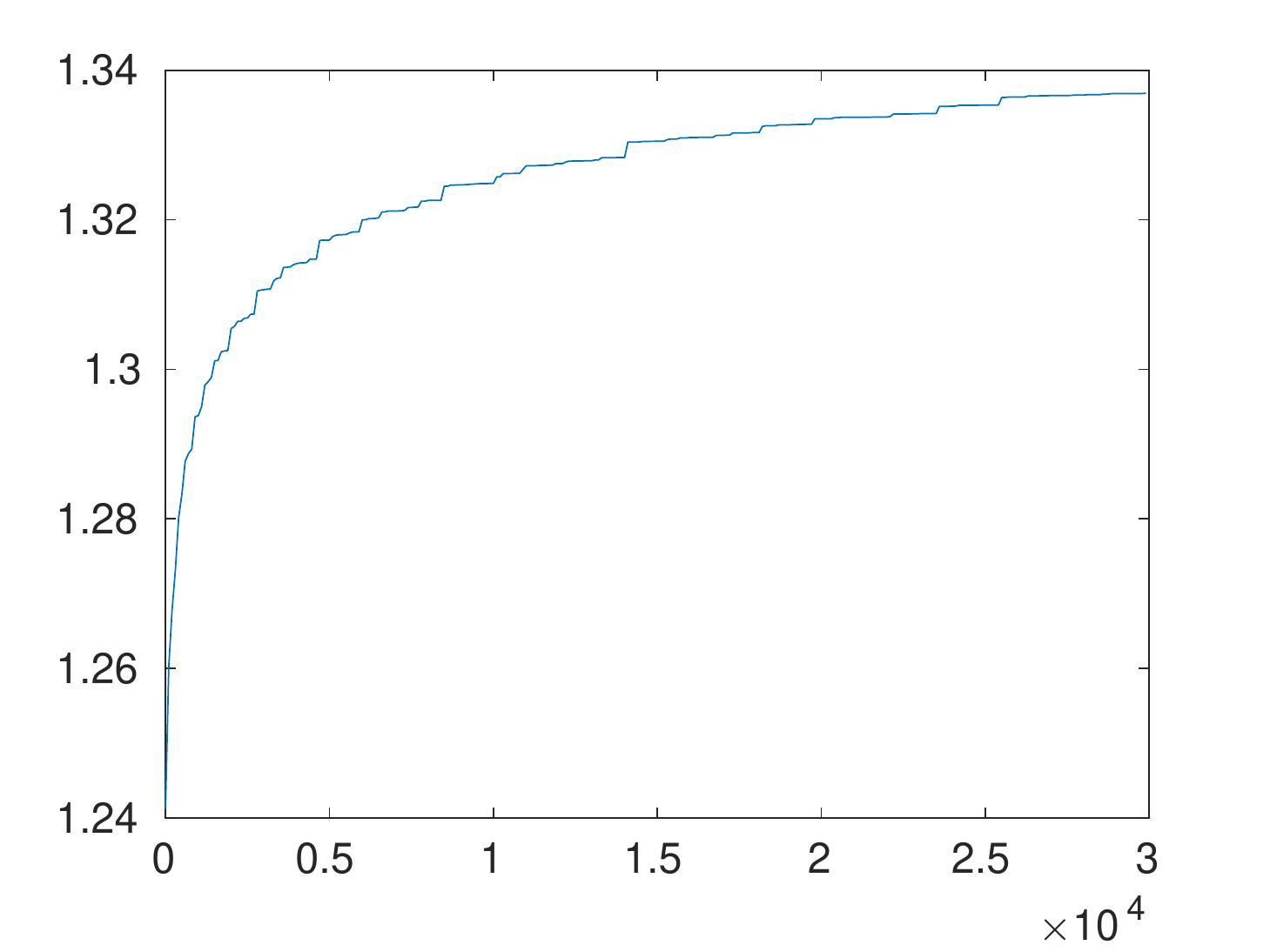}
\includegraphics[width=6.5cm]{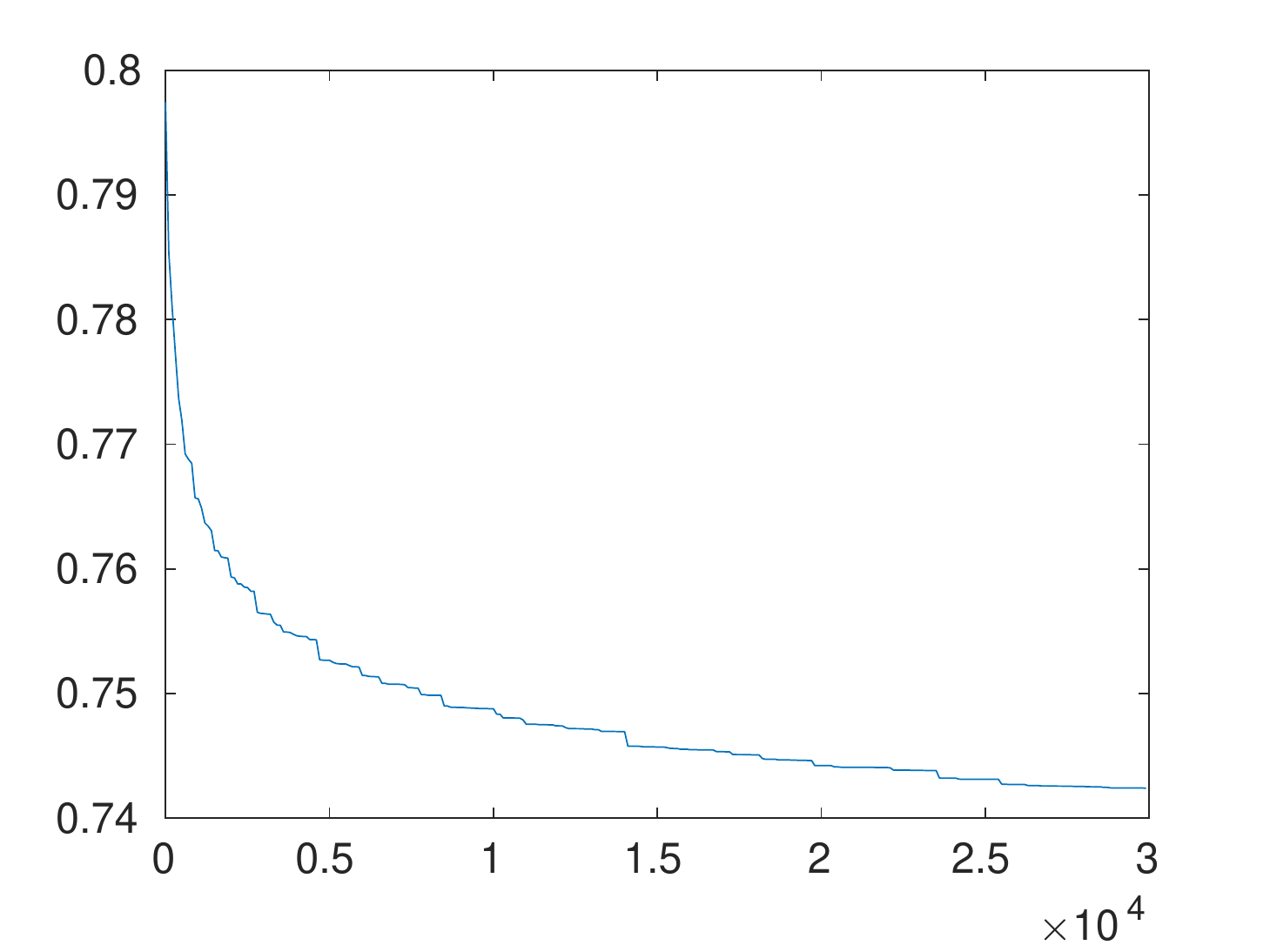} 
\caption{The largest (left) and the smallest (right) eigenvalue of $(\langle \CalC_i,\CalC_j\rangle)_{i,j=1}^N$ for $1\le N\le 3\cdot 10^4$.}\label{fig:Gram} 
\end{center}
\end{figure}
\end{enumerate}
\end{rem}
\emph{Proof of Theorem \ref{thm:CSd1}: Step 1.}\\
First we reformulate the definition of a (finite) Riesz sequence as an eigenvalue problem
of its Gram matrix. This reformulation is rather straightforward and by no means new, see \cite{KimLim} or \cite[Chapter 1.8]{Young}.
Let $H$ be a real Hilbert space and let $\{x_i\}_{i=1}^N\subset H$. Then, for every $\alpha=(\alpha_1,\dots,\alpha_N)^T\in\R^N$
\[
\left\|\sum_{i=1}^N \alpha_i x_i\right\|^2=\sum_{i,j=1}^N \alpha_i\alpha_j\langle x_i,x_j\rangle = \alpha^TG\alpha,
\]
where $G=(g_{i,j})_{i,j=1}^N$ with $g_{i,j}=\langle x_i,x_j\rangle$ is the Gram matrix of $\{x_i\}_{i=1}^N$.
Therefore, \eqref{eq:Riesz} is equivalent to  $A\alpha^T\alpha\le \alpha^TG\alpha\le B\alpha^T\alpha$ for every $\alpha\in\R^N$ 
or  simply to $\sigma(G)\subset [A,B].$
To show that this is indeed true for a given Gram matrix $G$, we will use the Gershgorin circle theorem \cite[Theorem 6.1.1]{HJ},
which states that
\[
\sigma(G)\subset \bigcup_{i=1}^N\left[g_{i,i}-\sum_{j\not=i}|g_{i,j}|,g_{i,i}+\sum_{j\not=i}|g_{i,j}|\right].
\]
\vskip.5cm
\emph{Proof of Theorem \ref{thm:CSd1}: Step 2.}\\
In this step, we show that $\{1\}\cup \{\sqrt{3}\,\CalC_k,\sqrt{3}\,\CalS_k:k\le N\}$ forms a Riesz sequence
for every $N\in\N$ with the Riesz constants independent on $N$. By Lemma \ref{lem:inner_productC},
its Gram matrix $G_N\in\R^{(2N+1)\times (2N+1)}$ is a block matrix with three blocks. The first one is just a $1\times 1$
block corresponding to the constant function, the second and the third block are $N\times N$ matrices of the inner products
$\big(3\langle \CalC_i,\CalC_j\rangle\big)_{i,j=1}^N$ and $\big(3\langle \CalS_i,\CalS_j\rangle\big)_{i,j=1}^N$, respectively, i.e., 

\newcommand{\bigzero}{\mbox{\normalfont\Large\bfseries 0}}
\newcommand{\rvline}{\hspace*{-\arraycolsep}\vline\hspace*{-\arraycolsep}}

\[
G_N=\begin{pmatrix}
  \begin{matrix}
 {1}
  \end{matrix}
  & \dashrule[-0.6ex]{0.4}{3 1.5 3 1.5 3} 
 & \bigzero 
 &\dashrule[-0.6ex]{0.4}{3 1.5 3 1.5 3} 
 & \bigzero\\
\hdashline[2pt/2pt]
  \bigzero & \dashrule[-9ex]{0.4}{3 1.5 3 1.5 3 1.5 3 1.5 3 1.5 3 1.5 3 1.5 3 1.5 3 1.5 3 1.5 3
  1.5 3 1.5 3 1.5 3 1.5 3 1.5 3 1.5 3 1.5 3 1.5 3 1.5 3 
  }  
  &    
 \begin{matrix}
  \overbrace{3\langle \CalC_1,\CalC_1\rangle}^{=1} & \ldots & 3\langle \CalC_1,\CalC_N\rangle\\
  \vdots & \ddots & \vdots \\
  3\langle \CalC_N,\CalC_1\rangle & \ldots & \underbrace{3\langle \CalC_N,\CalC_N\rangle}_{=1}
  \end{matrix} 
 & \dashrule[-9ex]{0.4}{3 1.5 3 1.5 3 1.5 3 1.5 3 1.5 3 1.5 3 1.5 3 1.5 3 1.5 3 1.5 3
  1.5 3 1.5 3 1.5 3 1.5 3 1.5 3 1.5 3 1.5 3 1.5 3 1.5 3 
  }  
  & \bigzero\\
  \hdashline[2pt/2pt]
  \bigzero & 
  \dashrule[-9ex]{0.4}{3 1.5 3 1.5 3 1.5 3 1.5 3 1.5 3 1.5 3 1.5 3 1.5 3 1.5 3 1.5 3
  1.5 3 1.5 3 1.5 3 1.5 3 1.5 3 1.5 3 1.5 3 1.5 3 1.5 3 
  } 
  & \bigzero & \dashrule[-9ex]{0.4}{3 1.5 3 1.5 3 1.5 3 1.5 3 1.5 3 1.5 3 1.5 3 1.5 3 1.5 3 1.5 3
  1.5 3 1.5 3 1.5 3 1.5 3 1.5 3 1.5 3 1.5 3 1.5 3 1.5 3 
  } 
  & 
  \begin{matrix}
  \overbrace{3\langle \CalS_1,\CalS_1\rangle}^{=1} & \ldots & 3\langle \CalS_1,\CalS_N\rangle\\
  \vdots & \ddots & \vdots \\
  3\langle \CalS_N,\CalS_1\rangle & \ldots & \underbrace{3\langle \CalS_N,\CalS_N\rangle}_{=1}
  \end{matrix} 
\end{pmatrix}.
\]

We apply the Gershgorin theorem to $G_N$.  Therefore, we need to estimate the row sums of $G_N$. 
For the first row we see that $g_{1,1} =1$ and $g_{1,j}=0$ for all $j\neq 1$. 

 Next we assume that $i\le N$ is an odd number, i.e. that
\[
i=q_1^{\alpha_1}\dots q_n^{\alpha_n}
\]
for some primes $q_1,\dots,q_n\ge 3$ and integers $\alpha_1,\dots,\alpha_n\ge 1$.
For the $(i+1)$-th row ($i=1,\ldots, N$, corresponding to $\CalC_i$)  we conclude  
\begin{equation}\label{eq:C_sum}
\langle \sqrt{3}\,{\CalC_i},1\rangle+\sum_{j=1}^N 3\cdot \langle{\mathcal C}_i,{\mathcal C}_j\rangle+\sum_{j=1}^N 3\cdot \langle{\mathcal C}_i,{\CalS}_j\rangle=\sum_{1\le j\le N, j\ \text{odd}} 3\cdot \langle{\mathcal C}_i,{\mathcal C}_j\rangle
\le \sum_{j\in\N, j\ \text{odd}} 3\cdot \langle{\mathcal C}_i,{\mathcal C}_j\rangle. 
\end{equation}
Every odd $j$ can be written as $j=q_1^{\beta_1}\dots q_n^{\beta_n}\cdot J$, where $\beta_1,\dots,\beta_n\ge 0$ and
$J$ is an odd integer, not divisible by any of $q_1,\dots,q_n$,
i.e., with $\gcd(J,i)=1$. Observe that with this notation
\[
\gcd(i,j)=\prod_{u=1}^n q_u^{\min(\alpha_u,\beta_u)}.
\]
Therefore, we can use Lemma \ref{lem:inner_productC} and rewrite \eqref{eq:C_sum} as
\begin{align*}
\sum_{j\in\N, \; j\text{ odd}} 3\cdot \langle{\mathcal C}_i,{\mathcal C}_j\rangle
&=\sum_{\beta_1,\dots,\beta_n=0}^\infty\sum_{\substack{J\in\N, J\ \text{odd}\\\gcd(J,i)=1}}
\frac{\displaystyle\prod_{u=1}^n q_u^{4\min(\alpha_u,\beta_u)}} {{\left(q_1^{\alpha_1}\cdots q_n^{\alpha_n}\right)^2}
\cdot \left(q_1^{\beta_1}\dots q_n^{\beta_n}\cdot J \right)^2}\\
&=\sum_{\substack{J\in\N, J\ \text{odd}\\\gcd(J,i)=1}} \frac{1}{J^2} \cdot \sum_{\beta_1=0}^\infty\frac{1}{\left[q_1^{\alpha_1+\beta_1-2\min(\alpha_1,\beta_1)}\right]^2}
\cdots \sum_{\beta_n=0}^\infty\frac{1}{\left[q_n^{\alpha_n+\beta_n-2\min(\alpha_n,\beta_n)}\right]^2}.
\end{align*}
Next, we simplify the individual terms.
\begin{align*}
\sum_{\substack{J\in\N, J\ \text{odd}\\\gcd(J,i)=1}} \frac{1}{J^2}&=
\prod_{\substack{p\ge 3: p\text{ prime}\\ p\not\in\{q_1,\dots,q_n\}}}\left(1+\frac{1}{p^2}+\frac{1}{p^4}+\dots\right)
=\prod_{\substack{p\ge 3: p\text{ prime}\\ p\not\in\{q_1,\dots,q_n\}}} \frac{1}{1-1/p^2}
\end{align*}
and
\begin{align*}
\sum_{\beta=0}^\infty\frac{1}{\left[q^{\alpha+\beta-2\min(\alpha,\beta)}\right]^2}&=\sum_{\beta=0}^\alpha\frac{1}{q^{2(\alpha-\beta)}}
+\sum_{\beta=\alpha+1}^\infty\frac{1}{q^{2(\beta-\alpha)}}\\
&\le \sum_{r=0}^\infty\frac{1}{q^{2r}}+\sum_{r=1}^\infty\frac{1}{q^{2r}}=\left(1+1/q^2\right)\cdot \frac{1}{1-1/q^2}.
\end{align*}
Therefore,
\begin{align}
\eqref{eq:C_sum}&\le \prod_{p\ge 3: p\text{ prime}}\frac{1}{1-1/p^2}\cdot\prod_{u=1}^n\left(1+\frac{1}{q_u^2}\right)
\le \prod_{p\ge 3: p\text{ prime}}\frac{1+1/p^2}{1-1/p^2}=\frac 32,  \label{euler-prod}
\end{align}
where in the last step we used the following Euler product \cite[Page 5]{Titch} attributed already to Ramanujan
\[
\prod_{p\text{ prime}}\frac{1+1/p^2}{1-1/p^2}=\frac 52. 
\]
Therefore, we get 
\begin{equation}\label{gersh-est}
\sum_{j\not= i}|(G_N)_{i,j}|=\sum_{j\not= i} 3\langle{\mathcal C}_i,{\mathcal C}_j\rangle
=\sum_{j= 1}^N 3\langle{\mathcal C}_i,{\mathcal C}_j\rangle-3\langle {\mathcal C_i,{\mathcal C}_i}\rangle\le \frac 32-1=\frac 12.
\end{equation}

If $i$ is even it follows from  Lemma \ref{lem:inner_productC}, assertions 2. and 3., that the estimates above remain the same since $3\cdot \langle{\mathcal C}_i,{\mathcal C}_j\rangle=\frac{\gcd(i,j)^4}{i^2\cdot j^2}\neq 0$ only for $j$'s with the same power of $2$
in their prime factorization as $i$, which then cancels out.

If we replace in \eqref{eq:C_sum}  $\CalC_i$  by  $\CalS_i$ (corresponding to the rows $(N+1)+i$, $i=1,\ldots, N$ of the Gram matrix), we obtain instead the estimate
\begin{equation}\label{eq:S_sum}
\langle \sqrt{3}\,{\CalS_i},1\rangle+\sum_{j=1}^N 3\cdot \langle{\mathcal S}_i,{\mathcal C}_j\rangle+\sum_{j=1}^N 3\cdot 
|\langle{\CalS}_i,{\CalS}_j\rangle|\leq \sum_{1\le j\le N, j\ \text{odd}} 3\cdot |\langle{\mathcal S}_i,{\mathcal S}_j\rangle|
\le \sum_{j\in\N, j\ \text{odd}} 3\cdot |\langle{\mathcal S}_i,{\mathcal S}_j\rangle|. 
\end{equation}
The term on the right hand side of \eqref{eq:S_sum} can be bounded by \eqref{euler-prod} as before and we again obtain
\begin{equation}\label{gersh-est-2}
\sum_{j\not= i}|(G_N)_{i,j}|=\sum_{j\not= i} 3|\langle{\mathcal S}_i,{\mathcal S}_j\rangle|
=\sum_{j= 1}^N 3|\langle{\mathcal S}_i,{\mathcal S}_j\rangle|-3\langle {\mathcal S_i,{\mathcal S}_i}\rangle\le \frac 32-1=\frac 12.
\end{equation}
By Gershgorin's theorem, we deduce   
$\sigma(G_N)\subset \{1\}\cup [\frac 12,\frac 32] =[\frac 12,\frac 32].$

\vskip.5cm
\emph{Proof of Theorem \ref{thm:CSd1}: Step 3.}\\
The third step of the proof of Theorem \ref{thm:CSd1}, i.e., the passage
to the limit $N\to\infty$, is quite standard and straightforward
(cf. \cite{KimLim} and \cite{Ole_P1,Ole_P2} for the so-called ``projection method'') and is contained in the following lemma.
\begin{lemma}
Let $H$ be a real Hilbert space and let $(x_n)_{n=1}^\infty\subset H$ be an infinite sequence. If $\{x_1,\dots,x_N\}$ is a Riesz sequence
for every $N\in\N$ with Riesz constants $A$ and $B$ independent on $N$, then $\{x_1,x_2,\dots\}$ is also a Riesz sequence with Riesz constants $A$ and $B$.
\end{lemma}
\begin{proof}
Let $\alpha=(\alpha_1,\alpha_2,\dots)$ be a square-summable sequence. Since 
\[
\left\|\sum_{n=n_0}^{n_1}\alpha_n x_n\right\|^2\le B \sum_{n=n_0}^{n_1}\alpha_n^2,
\]
the partial sums of $\sum_{n=1}^\infty \alpha_nx_n$ form a Cauchy sequence and therefore the series is convergent. Furthermore, by the triangle inequality
\[
\left\|\sum_{n=1}^{N}\alpha_n x_n\right\|\to \left\|\sum_{n=1}^{\infty}\alpha_n x_n\right\|.
\]
Hence, we can take the limit $N\to\infty$ in
\[
A \sum_{n=1}^{N}\alpha_n^2\le \left\|\sum_{n=1}^{N}\alpha_n x_n\right\|^2\le B \sum_{n=1}^{N}\alpha_n^2
\]
and the result follows.
\end{proof}

\vskip.5cm
\emph{Proof of Theorem \ref{thm:CSd1}: Step 4.}\\
As the last step, we show that ${\mathcal R}_1$ is not only a Riesz sequence but also a Riesz basis,
i.e., that its closed linear span is the whole space $L_2([0,1])$. We rely on the fact that the trigonometric
system
\begin{equation}\label{eq:T1}
{\mathcal T}_1:=\{1\}\cup \{\sqrt{2}\cos(2\pi kx):k\in\N\}\cup\{\sqrt{2}\sin(2k\pi x):k\in\N\}
\end{equation}
forms an orthonormal basis of $L_2([0,1])$. We show that every function from \eqref{eq:T1}
lies in the closed linear span of ${\mathcal R}_1$ and, therefore, the closed linear span of ${\mathcal T}_1$
is contained in the closed linear span of ${\mathcal R}_1.$

We start with the following lemma, which gives an explicit decomposition of $\cos(2\pi x)$ and 
$\sin(2\pi x)$ in ${\mathcal R}_1$. Its statement requires the notion of the M\"obius function,
which is defined for every positive integer $n\in\N$ as
\[
\mu(n)=\begin{cases}
+1,\quad &\text{if $n$ is a square-free integer with an even number of prime factors},\\
-1,\quad &\text{if $n$ is a square-free integer with an odd number of prime factors},\\
0,\quad &\text{if $n$ is not a square-free integer, i.e., if it is divisible by some squared prime.}
\end{cases}
\]

\begin{lemma} For every $l\in\N$, let 
$\overline{\CalC}_l(x)=\sqrt{3}\,\CalC_l(x)/\mu$ and $\overline{\CalS}_l(x)=\sqrt{3}\,\CalS_l(x)/\mu$, 
where $\mu$ is the constant from \eqref{eq:mu}. Then
\begin{equation}\label{eq:decomp_c}
c_1(x):=\sqrt{2}\cos(2\pi x)=\sum_{l=0}^\infty \frac{\mu(2l+1)}{(2l+1)^2}\,\overline{\CalC_{2l+1}}(x)
\end{equation}
and
\begin{equation}\label{eq:decomp_s}
s_1(x):=\sqrt{2}\sin(2\pi x)=\sum_{l=0}^\infty (-1)^l\frac{\mu(2l+1)}{(2l+1)^2}\,\overline{\CalS_{2l+1}}(x)
\end{equation}
with the convergence being in $L_2([0,1])$.
\end{lemma}
\begin{proof}
We first reformulate \eqref{eq:CSFourier} as
\begin{equation}\label{eq:inv:1}
\overline{\CalC}_1(x)=\frac{\sqrt{3}\, {\mathcal C}(x)}{\mu}
=\sum_{m=1}^\infty\frac{\alpha_{m}}{m^2}\cdot c_{m}(x),
\end{equation}
where $\alpha_{2m+1}=1$ and $\alpha_{2m}=0$. Note, that \eqref{eq:inv:1} converges in $L_2([0,1])$.

We show that there is a unique bounded sequence $(\beta_l)_{l=1}^\infty$, such that
\begin{equation}\label{eq:decompc1}
c_1(x)=\sum_{l=1}^\infty\frac{\beta_l}{l^2}\cdot {\overline{\CalC}_l}(x)
\end{equation}
with the convergence in $L_2([0,1])$ and that 
\begin{equation}\label{eq:inv:3}
    \beta_l=\begin{cases}0& \text{for $l\in\N$ even},\\
\mu(l)&\text{for $l\in\N$ odd}. 
\end{cases}
\end{equation}

Using \eqref{eq:inv:1}, we observe that \eqref{eq:decompc1} holds for bounded sequence $(\beta_l)_{l=1}^\infty$ if, and only if,
\begin{equation}\label{eq:inv:2}
c_1(x)=\sum_{l=1}^\infty\sum_{m=1}^\infty\frac{\beta_l\cdot \alpha_{m}}{l^2\cdot m^2} c_{lm}(x).
\end{equation}
We compare the coefficients of $c_k(x)$ on both sides of \eqref{eq:inv:2} and observe that \eqref{eq:inv:2}
is equivalent to the system of equations
\begin{align}
\begin{split}
1&=\beta_1\alpha_1,\label{eq:inv:4}\\
0&=\sum_{(l,m):l\cdot m=n}\beta_l\alpha_m,\quad n=2,3,\dots.
\end{split}
\end{align}

This system could be solved by using the M\"obius inversion formula \cite[p.~3]{Titch}, but one can also proceed directly.
From $1=\beta_1\alpha_1$ we obtain $\beta_1=1$ and from $\beta_1\alpha_2+\beta_2\alpha_1=0$ we get 
$\beta_2=0$. We show by induction that $\beta_{2n}=0$ also for all $n\ge 1$. Let this be true
for all integers smaller than $n$. Then
\[
0=\sum_{(l,m):l\cdot m=2n}\beta_l\alpha_m=
\sum_{(l,m):l\cdot m=n}\beta_{2l}\alpha_m+\sum_{(l,m):l\cdot m=n}\beta_l\alpha_{2m}=\beta_{2n}\alpha_1
\]
gives $\beta_{2n}=0$ as well.
Similarly, from $0=\beta_p\alpha_1+\beta_1\alpha_p$ we get $\beta_p=-1$
for every prime $p\neq 2$.

Let now $p=p_1p_2$ with odd primes $p_1,p_2$. Then $\beta_p=1$ follows from 
\begin{equation*}
0=\beta_p\ \alpha_1 + \beta_{p_1}\alpha_{p_2}+\beta_{p_2}\alpha_{p_1}+ \beta_1\alpha_p
= \beta_p\cdot 1+(-1\cdot 1)+(-1\cdot 1)+1.
\end{equation*}
The formula for a general $p=p_1\cdot \ldots \cdot p_k$, with distinct odd primes $p_i$ follows by induction.
Observe that $p$ is divisible by all $p_1^{e_1}\cdot\ldots\cdot p_k^{e_k}$ with $e=(e_1,\dots,e_k)\in\{0,1\}^k$. Hence
\begin{align*}
0=\sum_{(l,m):l\cdot m=p}\beta_l\alpha_m=\sum_{l|p}\beta_l=\beta_p+\sum_{j=0}^{k-1}\binom{k}{j}(-1)^j
=\beta_p+\sum_{j=0}^{k}\binom{k}{j}(-1)^j-(-1)^k=\beta_p-(-1)^k.
\end{align*}
We conclude, that $\beta_n=\mu(n)$ for every odd square-free integer $n$.

Finally, we show that $\beta_n=0$ for every positive odd integer $n$, which is not square-free.
Therefore, we assume that $n=p_1^u\cdot q$, where $p_1$ is an odd prime, $q$ is an odd integer not divisible by $p_1$, $u\ge 2$
and that the statement is true for all integers smaller than $n$. We obtain
\begin{align*}
    0=\sum_{(l,m):l\cdot m=n}\beta_l\alpha_m=\sum_{l|p_1^uq}\beta_l=\sum_{l|q}(\beta_l+\beta_{p_1l}+\dots+\beta_{p_1^ul}).
\end{align*}
If $l<q$ is square-free, then the first two terms in this sum have values $+1$ and $-1$, respectively, and the others vanish
by the induction assumption. If $l<q$ is not square-free, then all the terms vanish again by assumption. Finally, if $l=q$
the same argument applies leaving us with $\beta_n=0$.
We conclude that the sequence $(\beta_l)_{l=1}^\infty$ given by \eqref{eq:inv:3} indeed satisfies the system \eqref{eq:inv:4}
which in turn gives \eqref{eq:decompc1}.

Finally, \eqref{eq:decomp_s} follows from \eqref{eq:decomp_c} using the simple relation $\CalS(x)=\CalC(x-1/4).$ 
The factor $(-1)^l$ results from the relation
\[
\CalC_{2l+1}\Big(x-\frac 14\Big)=(-1)^l\CalC\Big((2l+1)x-\frac 14\Big)=(-1)^l \CalS_{2l+1}(x). 
\]
\end{proof}

\section{Multivariate case}
\label{sect-multivariate}

The main aim of this section is to generalize Theorem \ref{thm:CSd1} to higher dimensions $d\ge 1$
and to provide a Riesz basis of $L_2([0,1]^d)$, which is easily expressed by artificial neural networks
with $\ReLU$ activation function.
The most natural approach would be to consider the tensor products of the functions from ${\mathcal R}_1$, i.e., a system of functions of the form $(x,y)\to \CalC_k(x)\cdot \CalC_l(y)$ etc. 
Indeed, it is quite easy to show that tensor products of elements of a Riesz sequence form again a Riesz sequence \cite{Bour}.
This approach is quite classical in analysis and there exist many multivariate bases and systems with a tensor product structure.
Unfortunately, the Riesz constants of the tensor product system are in general given as products of the Riesz constants
of the univariate Riesz sequences, cf. \cite[Theorem 4.1]{Bour}. Applying the tensor product construction to ${\mathcal R}_1$
would therefore lead to an exponential dependence of the ratio of the Riesz constants on the dimension.

Furthermore, the tensor product approach does not fit really well to artificial neural networks.
The reason is that it is surprisingly difficult to construct a neural network, which for two real inputs $x$ and $y$
outputs the product $xy$ (or at least its approximation). In general, one first approximates the square function $t\to t^2$
and then applies the formula $xy=[(x+y)^2-(x-y)^2]/4.$ We refer to \cite{EPGB, Schm, Yarot} for details. Therefore, we are looking
for another multivariate Riesz basis of piecewise affine functions, which can  be constructed without the use of (tensor) products, 
but where inner products with fixed vectors in $\R^d$ are allowed.

 Before we state our results, we need some additional notation. If $\alpha=(\alpha_1,\dots,\alpha_d)\in\Z^d$, we say that 
$\alpha \righttriangleplus 0$  if the first non-zero  entry  of $\alpha$ is positive.
The multivariate analogue of ${\mathcal R}_1$ is then defined simply as
\[
{\mathcal R}_d:=\{1\}\cup\{{\mathcal C}(\alpha\cdot x):\alpha\righttriangleplus 0\}\cup\{\CalS(\alpha\cdot x):\alpha\righttriangleplus 0\},
\]
where we interpret the functions $\CalC$ and $\CalS$ as $1-$periodic functions on $\R$.
Note that we need to restrict ourselves to indices $\alpha \righttriangleplus 0$ here,
since ${\mathcal C}(\alpha\cdot x)={\mathcal C}(-\alpha\cdot x)$ and 
${\mathcal S}(\alpha\cdot x)=-{\mathcal S}(-\alpha\cdot x)$, respectively.

Furthermore, we say that two non-zero $\alpha,\beta\in\R^d$ are co-linear if there is $t\not=0$ such that $\alpha=t\beta.$
Obviously, if $\alpha,\beta\in\Z^d$ with $\alpha,\beta\righttriangleplus 0$ are co-linear, then $t>0$. In the rest of this section,
the inner product $\displaystyle\langle f, g\rangle=\int_{[0,1]^d}f(x)g(x)dx$ denotes the inner product in $L_2([0,1]^d)$,
the space of real square integrable functions on $[0,1]^d$.

The multivariate analogue of Lemma \ref{lem:inner_productC}, which characterizes the inner products of the elements of
${\mathcal R}_d$ then looks as follows.
\begin{lemma}\label{lem:IPd}
Let $\alpha,\beta\in\Z^d$ with $\alpha,\beta\righttriangleplus 0$. Then
\begin{enumerate}
\item
$\langle\CalC(\alpha\cdot x),\CalS(\beta\cdot x)\rangle=0$.
\item $\langle {\mathcal C}(\alpha\cdot x),{\mathcal C}(\beta\cdot x)\rangle=\langle {\mathcal S}(\alpha\cdot x),{\mathcal S}(\beta\cdot x)\rangle=0$
if $\alpha$ and $\beta$ are not co-linear, or if $\alpha=t\beta$, but $t$ can not be written as a ratio of two odd positive integers.
\item If $\displaystyle \alpha=\frac{2p+1}{2q+1}\cdot \beta$ with coprime integers $2p+1$ and $2q+1$ (i.e.,   $\gcd(2p+1, 2q+1)=1$), then
\begin{equation}
3\langle {\mathcal C}(\alpha\cdot x),{\mathcal C}(\beta\cdot x)\rangle=
\frac{1}{(2p+1)^2(2q+1)^2}
\label{eq:CSinnprod}
\end{equation}
and
\begin{equation}\label{eq:SSinnprod}
3\langle\CalS(\alpha\cdot x),\CalS(\beta\cdot x)\rangle=\frac{(-1)^{p+q}}{(2p+1)^2(2q+1)^2}.
\end{equation}
\end{enumerate}
\end{lemma}

\begin{rem}
Lemma \ref{lem:IPd} includes  Lemma \ref{lem:inner_productC} as a special case.  In particular,  if  $d=1$
 then $\alpha$ and $\beta$ are always co-linear. Moreover, if the prime factorizations of $\alpha$ and $\beta$ contain 
 different powers of $2$, then we have  $\alpha\neq \frac{2p+1}{2q+1}\beta$ for all integers $p,q$. 
\end{rem}

\begin{proof}[Proof of Lemma \ref{lem:IPd}]
First, we recall the elementary formulas
\begin{align}
\notag\cos(u_1+\dots+u_d)&=\operatorname{Re}(e^{i(u_1+\dots+u_d)})=\operatorname{Re}\left[\prod_{j=1}^d(\cos u_j+i\sin u_j)\right]\\
\label{eq:cossum}&=\sum_{J\subset\{1,\dots,d\}, \#J\ \text{even}}(-1)^{\#J/2}\prod_{j\not\in J}\cos u_j\cdot\prod_{j\in J}\sin u_j
\end{align}
and
\begin{align}
\notag\sin(u_1+\dots+u_d)&=\operatorname{Im}(e^{i(u_1+\dots+u_d)})=\operatorname{Im}\left[\prod_{j=1}^d(\cos u_j+i\sin u_j)\right]\\
\label{eq:sinsum}&=\sum_{J\subset\{1,\dots,d\}, \#J\ \text{odd}}(-1)^{(\#J-1)/2}\prod_{j\not\in J}\cos u_j\cdot\prod_{j\in J}\sin u_j.
\end{align}

Next, we proceed to the proof of \eqref{eq:CSinnprod}.
Let us fix $\alpha,\beta\righttriangleplus 0$. We apply \eqref{eq:CSFourier} followed by \eqref{eq:cossum} and obtain
\begin{align}
\notag 3&\langle {\mathcal C}(\alpha\cdot x),{\mathcal C}(\beta\cdot x)\rangle=\sum_{m,n\ge0}\frac{\mu^2}{(2m+1)^2(2n+1)^2}\int_{[0,1]^d}c_{2m+1}(\alpha\cdot x)c_{2n+1}(\beta\cdot x)dx\\
\notag&=\sum_{m,n\ge0}\sum_{\substack{J,K\subset \{1,\dots,d\}\\\#J,\#K\ \text{even}}}\frac{2\mu^2}{(2m+1)^2(2n+1)^2}\cdot (-1)^{(\#J+\#K)/2}\cdot\\
\notag&\qquad\qquad\cdot\frac{1}{2^{d}}\cdot \int_{[0,1]^d}\prod_{j\not\in J}c_{2m+1}(\alpha_jx_j)\prod_{j\in J}s_{2m+1}(\alpha_jx_j)\prod_{k\not\in K}c_{2n+1}(\beta_kx_k)\prod_{k\in K}s_{2n+1}(\beta_kx_k)dx\\
\notag&=\sum_{m,n\ge0}\sum_{\substack{J\subset \{1,\dots,d\}\\\#J\ \text{even}}}\frac{2\mu^2}{(2m+1)^2(2n+1)^2}\\
\label{eq:product}&\qquad\qquad\cdot\frac{1}{2^{d}}\cdot \prod_{j\not\in J}\int_{[0,1]}c_{2m+1}(\alpha_jx_j)c_{2n+1}(\beta_jx_j)dx_j\cdot\prod_{j\in J}\int_{[0,1]}s_{2m+1}(\alpha_jx_j)s_{2n+1}(\beta_jx_j)dx_j.
\end{align}
Next, we discuss, when the last product vanishes for given $J$. 
First, this happens if $\alpha_j=0$ or $\beta_j=0$ for any $j\in J$. If $j\in J$ and both $\alpha_j$
and $\beta_j$ are non-zero, then the last product vanishes also if $(2m+1)|\alpha_j|\not=(2n+1)|\beta_j|$.
And finally, the product is zero also if $j\not \in J$ and $(2m+1)|\alpha_j|\not=(2n+1)|\beta_j|$.

Equivalently, \eqref{eq:product} is not equal to zero if $(2m+1)|\alpha_j|=(2n+1)|\beta_j|$ for every $j\in\{1,\dots,d\}$
and $\alpha_j$ and $\beta_j$ are different from zero if $j\in J$.
If we denote $|\alpha|=(|\alpha_1|,\dots,|\alpha_n|)$ (and similarly for $|\beta|$), we will therefore
restrict ourselves for the rest of the proof
to $m,n\ge 0$ with
\begin{equation}\label{eq:alphabeta}
(2m+1)|\alpha|=(2n+1)|\beta|.
\end{equation}
If there is no pair of integers $(m,n)\in\N_0^2$, such that \eqref{eq:alphabeta} holds,
then $\langle {\mathcal C}(\alpha\cdot x),{\mathcal C}(\beta\cdot x)\rangle=0$.
Furthermore, we may consider only sets $J\subset\{1,\dots,d\}$ with an even number of elements, which are
subsets of $\supp(\alpha)=\supp(\beta).$

If \eqref{eq:alphabeta} holds, than the univariate integrals in \eqref{eq:product} are equal to one
for $j\not\in J$ and $j\in\supp(\alpha)$. They are equal to two, if $j\not\in J$ and $j\not\in\supp(\alpha)$.
And if $j\in J\subset\supp(\alpha)$, then the integral is +1 if $(2m+1)\alpha_j=(2n+1)\beta_j$ and it is equal to $-1$ if $(2m+1)\alpha_j=-(2n+1)\beta_j$.

We denote $D=\{j:\sign(\alpha_j)\cdot\sign(\beta_j)=-1\}\subset\supp(\alpha)=\supp(\beta)$,
$\nu=\#D$ and $n=\#\supp(\alpha)\ge \nu$.
Using this notation, we obtain
\begin{align*}
3\langle {\mathcal C}(\alpha\cdot x),{\mathcal C}(\beta\cdot x)\rangle&=
\sum_{\substack{m,n\ge0 \\ (2m+1)|\alpha|=(2n+1)|\beta|}}\frac{2\mu^2}{(2m+1)^2(2n+1)^2\cdot 2^{n}}
\sum_{\substack{J\subset \supp(\alpha)\\\#J\ \text{even}}}(-1)^{\#(J\cap D)}.
\end{align*}
If $\nu\ge 1$, we calculate
\begin{align*}
\sum_{\substack{J\subset \supp(\alpha)\\\#J\ \text{even}}}(-1)^{\#(J\cap D)}
&=\sum_{\substack{0\le a \le n-\nu\\ 0\le b\le \nu\\ a+b \text{\ even}}}(-1)^b \binom{n-\nu}{a}\binom{\nu}{b}\\
&=\sum_{b=0}^\nu (-1)^b\binom{\nu}{b}\sum_{\substack{0\le a \le n-\nu\\a+b \text{\ even}}}\binom{n-\nu}{a}=0, 
\end{align*}
where the last step follows since $ \sum_{b=0}^\nu (-1)^b\binom{\nu}{b}=(1-1)^{\nu}=0$. 
Hence, $\langle {\mathcal C}(\alpha\cdot x),{\mathcal C}(\beta\cdot x)\rangle=0$ if there exists
$1\le j\le d$ with $\sign(\alpha_j)\cdot \sign(\beta_j)=-1$ and we arrive at
\[
3\langle {\mathcal C}(\alpha\cdot x),{\mathcal C}(\beta\cdot x)\rangle=\sum_{\substack{m,n\ge0 \\ (2m+1)\alpha=(2n+1)\beta}}\frac{\mu^2}{(2m+1)^2(2n+1)^2}.
\]
The last sum is empty if $\alpha$ and $\beta$ are not co-linear or if we can not write $\alpha=t\beta$, where $t>0$ is a ratio
of two odd integers. Therefore, we assume that $\displaystyle \alpha=\frac{2p+1}{2q+1}\beta$ with $p,q\in\N_0$ and
that $2p+1$ and $2q+1$ are coprime integers. All pairs $(m,n)\in\N_0^2$ with $(2m+1)\alpha=(2n+1)\beta$ are then of the form
\begin{equation}\label{eq:mnl}
2m+1=(2q+1)(2l+1)\quad\text{and}\quad 2n+1=(2p+1)(2l+1),\quad l\in\N_0.
\end{equation}
This finally leads to
\[
3\langle {\mathcal C}(\alpha\cdot x),{\mathcal C}(\beta\cdot x)\rangle=\sum_{l=0}^\infty\frac{\mu^2}{(2q+1)^2(2p+1)^2(2l+1)^4},
\]
which combined with \eqref{eq:mu} gives \eqref{eq:CSinnprod}. As a byproduct, we also showed that
$\langle {\mathcal C}(\alpha\cdot x),{\mathcal C}(\beta\cdot x)\rangle=0$
if $\alpha$ and $\beta$ are not co-linear with a real factor $t$, which can be written as a ration of two odd positive integers.

Applying the same idea to the inner product of $\CalC(\alpha\cdot x)$ and $\CalS(\beta\cdot x)$,
we get a double sum over $J,K\subset\{1,\dots,d\}$
with $\#J$ even and $\#K$ odd. Therefore, it is not possible to match the univariate integrands
and their product always vanishes.
Finally, \eqref{eq:SSinnprod} follows in the same way, the only essential difference being the $(-1)^{m+n}$ factor
coming from \eqref{eq:CSFourier}. And an easy observation shows that under \eqref{eq:mnl}, the parity of $m+n$ is the same as the one of $p+q.$
\end{proof}

We complement Lemma \ref{lem:IPd} by the simple observation that the constant function is orthogonal
to all other elements of ${\mathcal R}_d$. The multivariate analogue of Theorem \ref{thm:CSd1} then reads as follows. 

\begin{thm}\label{thm:CSd}
Let $d\ge 1$. Then the system ${\mathcal R}_d$ forms a Riesz basis of $L_2([0,1]^d)$ with the Riesz constants independent of $d$.
To be more specific, the Riesz constants of the normalized system
\begin{equation}\label{eq:Rdnorm}
    \{1\}\cup\{\sqrt{3}\,{\mathcal C}(\alpha\cdot x):\alpha\righttriangleplus 0\}\cup\{\sqrt{3}\,\CalS(\alpha\cdot x):\alpha\righttriangleplus 0\}
\end{equation}
can be chosen as $A=1/2$ and $B=3/2$ independently of $d$.
\end{thm}
\begin{proof}

\emph{Step 1.}
Using Lemma \ref{lem:IPd} together with the Gershgorin circle theorem, it is surprisingly simple
to prove Theorem \ref{thm:CSd} with slightly worse constants $A$ and $B$, cf. Remark \ref{rem:badAB}.
To improve the Riesz constants to $A=1/2$ and $B=3/2$, we proceed more carefully.
Let us denote by ${\mathbb P}$ the set of primes and by ${\mathbb P}'={\mathbb P}\setminus\{2\}$
the set of odd primes. Then every odd $p\in\N$ can be written as $p=p_1^{k_1}\cdot\ldots\cdot p_n^{k_n}$
with $p_1,\dots,p_n\in {\mathbb P}'$ and $k_1,\dots,k_n\in\N$. If $p=1$, then we choose $n=0$ and interprete the empty product as one.

We use the following observation. To a fixed $\alpha\righttriangleplus 0$ and a pair of odd coprimed integers $2p+1$ and $2q+1$,
there exists at most one $\beta\righttriangleplus 0$ such that $\alpha=(2p+1)/(2q+1)\cdot\beta$. Then we obtain
\begin{align*}
3\sum_{\beta\righttriangleplus 0}\langle\CalC(\alpha\cdot x),\CalC(\beta\cdot x)\rangle&\le
\sum_{2p+1,2q+1\ \text{coprimes}}\frac{1}{(2p+1)^2(2q+1)^2}=\sum_{p,q\ge 1\ \text{odd coprimes}}\frac{1}{p^2q^2}\\
&=\sum_{\{p_1,\dots,p_n\}\subset {\mathbb P}'}\sum_{k_1,\dots,k_n=1}^\infty \frac{1}{p_1^{2k_1}\dots p_n^{2k_n}}
\sum_{\substack{q\ge 1, q\text{ odd}\\\gcd(q,p_1\dots p_n)=1}}\frac{1}{q^2}\\
&=\sum_{\{p_1,\dots,p_n\}\subset {\mathbb P}'}\prod_{u=1}^n \left(\sum_{k_u=1}^\infty\frac{1}{p_u^{2k_u}}\right)
\prod_{p\in {\mathbb P}'\setminus\{p_1,\dots,p_n\}}\left(\sum_{k=0}^\infty\frac{1}{p^{2k}}\right)\\
&=\sum_{\{p_1,\dots,p_n\}\subset {\mathbb P}'}\prod_{u=1}^n \frac{1}{p_u^2(1-1/p_u^{2})}
\prod_{p\in {\mathbb P}'\setminus\{p_1,\dots,p_n\}}\frac{1}{1-1/p^2}\\
&=\sum_{\{p_1,\dots,p_n\}\subset {\mathbb P}'}\frac{1}{p_1^2 \dots p_n^2}\prod_{p\in {\mathbb P}'}\frac{1}{1-1/p^2}\\
&=\prod_{p\in {\mathbb P}'}\frac{1}{1-1/p^2}\prod_{p\in {\mathbb P}'}(1+1/p^2)=\frac{3}{2}, 
\end{align*}
where the last step follows from  \eqref{euler-prod}.
Using the Gershgorin circle theorem in the same way as in the proof of Theorem \ref{thm:CSd1}
then gives the bounds $1/2\le A\le B\le 3/2$, independent of $d$.

\emph{Step 2.} We show that ${\mathcal R}_d$ is also a Riesz basis. The system
\begin{equation}\label{eq:ONB1}
\{1\}\cup\{c_k(x):k\in\N\}\cup\{s_k(x):k\in\N\}
\end{equation}
with $c_k(x)=\sqrt{2}\cos(2\pi kx)$ and $s_k(x)=\sqrt{2}\sin(2\pi kx)$ is an orthonormal basis of $L_2([0,1])$.
Therefore, all  possible tensor products of the functions from \eqref{eq:ONB1} form an orthonormal basis of $L_2([0,1]^d)$.
For this system we use the following notation
\begin{equation}\label{eq:ONB2}
\left\{\prod_{l\in L}c_{k_l}(x_l)\prod_{m\in M}s_{k_m}(x_m):L,M\subset\{1,\dots,d\}\ \text{disjoint}, k_l,k_m \in\N\right\}.
\end{equation}
Now we show that every function from \eqref{eq:ONB2} can be found in the closed linear span of ${\mathcal R}_d$. This will imply 
the completeness of ${\mathcal R}_d.$ First, we again recall two simple formulas
\begin{align*}
\prod_{u=1}^n \cos(\varphi_u)&=\prod_{u=1}^n \frac{e^{i\varphi_u}+e^{-i\varphi_u}}{2}=
\frac{1}{2^n}\sum_{e\in\{-1,+1\}^n}\exp(i[e_1\varphi_1+\dots+e_n\varphi_n])\\
&=\frac{1}{2^n}\sum_{e\in\{-1,+1\}^n}\cos(e_1\varphi_1+\dots+e_n\varphi_n) 
\end{align*}
and, similarly,
\begin{align*}
\prod_{u=1}^n \sin(\varphi_u)=\frac{(-1)^{\lfloor \frac{n}{2}\rfloor}}{2^n}\cdot
\begin{cases}
\displaystyle \sum_{e\in\{-1,+1\}^n}\cos(e_1\varphi_1+\dots+e_n\varphi_n)\cdot\prod_{j=1}^n e_j\ \text{if $n$ is even},\\
\displaystyle \sum_{e\in\{-1,+1\}^n}\sin(e_1\varphi_1+\dots+e_n\varphi_n)\cdot\prod_{j=1}^n e_j\ \text{if $n$ is odd.}
\end{cases}
\end{align*}
If $\#M$ is even,  we use the elementary property $\cos(\alpha)\cos(\beta)=(\cos(\alpha+\beta)+\cos(\alpha-\beta))/2$
and obtain
\begin{align}
\prod_{l\in L}&c_{k_l}(x_l)\prod_{m\in M}s_{k_m}(x_m)=2^{(\#L+\#M)/2}\prod_{l\in L}\cos(2\pi k_lx_l)\prod_{m\in M}\sin(2\pi k_mx_m)\notag\\
&=\frac{(-1)^{\lfloor\#M/2 \rfloor}}{2^{(\#L+\#M)/2}}\sum_{e\in\{-1,+1\}^{\#L}}\cos\left(\sum_{l\in L}e_l\cdot 2\pi k_l x_l\right)
\cdot \sum_{e'\in\{-1,+1\}^{\#M}}\cos \left(\sum_{m\in M}e'_m\cdot 2\pi k_mx_m\right)\cdot \prod_{m\in M}e'_m\notag\\
&=\frac{(-1)^{\lfloor\#M/2 \rfloor}}{2^{(\#L+\#M)/2}}\sum_{e\in\{-1,+1\}^{\#L+\#M}}\cos\left(\sum_{l\in L}e_l\cdot 2\pi k_l x_l\right)
\cdot\cos \left(\sum_{m\in M}e_m\cdot 2\pi k_mx_m\right)\cdot \prod_{m\in M}e_m\notag\\
&=\frac{(-1)^{\lfloor\#M/2 \rfloor}}{2\cdot 2^{(\#L+\#M)/2}}\sum_{e\in\{-1,+1\}^{\#L+\#M}}
\cos\left(\sum_{l\in L}e_l\cdot 2\pi k_l x_l+\sum_{m\in M}e_m\cdot 2\pi k_mx_m\right)
\cdot \prod_{m\in M}e_m\notag\\
&\qquad+\frac{(-1)^{\lfloor\#M/2 \rfloor}}{2\cdot 2^{(\#L+\#M)/2}}\sum_{e\in\{-1,+1\}^{\#L+\#M}}
\cos \left(\sum_{l\in L}e_l\cdot 2\pi k_l x_l-\sum_{m\in M}e_m\cdot 2\pi k_mx_m\right)
\cdot \prod_{m\in M}e_m\notag\\
&=\frac{(-1)^{\lfloor\#M/2 \rfloor}}{2^{(\#L+\#M)/2}}\sum_{e\in\{-1,+1\}^{\#L+\#M}}
\cos\left(\sum_{u\in L\cup M}e_u\cdot 2\pi k_u x_u\right)
\cdot \prod_{m\in M}e_m\notag\\
&=\frac{(-1)^{\lfloor\#M/2 \rfloor}}{2^{(\#L+\#M)/2}}\sum_{e\in\{-1,+1\}^{\#L+\#M}}
\sum_{l=0}^\infty \frac{\mu(2l+1)}{(2l+1)^2}\frac{\sqrt{3}}{\sqrt{2}\mu}
{\mathcal C}_{{2l+1}}\left(\sum_{u\in L\cup M}e_u k_u x_u\right)
\cdot \prod_{m\in M}e_m, \label{cos-cos-1}
\end{align}
where in the last step we used the Fourier decomposition from \eqref{eq:decomp_c}. 
If $\#M$ is odd,   we use instead the formula $\cos(\alpha)\sin(\beta)=(\sin(\alpha+\beta)+\sin(\beta-\alpha))/2$, which yields 
\begin{align*}
\prod_{l\in L}&c_{k_l}(x_l)\prod_{m\in M}s_{k_m}(x_m)\\
& =\frac{(-1)^{\lfloor\#M/2 \rfloor}}{2^{(\#L+\#M)/2}}\sum_{e\in\{-1,+1\}^{\#L+\#M}}
\sum_{l=0}^\infty (-1)^l\frac{\mu(2l+1)}{(2l+1)^2}\frac{\sqrt{3}}{\sqrt{2}\mu}
{\mathcal S}_{2l+1}\left(\sum_{u\in L\cup M}e_u k_u x_u\right)
\cdot \prod_{m\in M}e_m. 
\end{align*}

We interprete these formulas as a decomposition of a basis function from \eqref{eq:ONB2}
into ${\mathcal R}_d$, which converges in $L_2([0,1]^d)$. Reasoning similarly as in the proof
of Theorem \ref{thm:CSd1}, this  finishes the argument. 
\end{proof}

\begin{rem}\label{rem:badAB}
We observe that Lemma \ref{lem:IPd} implies for fixed $\alpha\righttriangleplus 0$
\begin{align*}
3\sum_{\beta\not=\alpha}|\langle\CalC(\alpha\cdot x),\CalC(\beta\cdot x)\rangle|\le
\sum_{p,q\ge 0}\frac{1}{(2p+1)^2(2q+1)^2}-1=\frac{\pi^4}{8^2}-1=0.522....
\end{align*}
Using Gershgorin's theorem similarly as in the proof of Theorem \ref{thm:CSd1}, we could have obtained
quite easily that \eqref{eq:Rdnorm} is a Riesz sequence with constants $A=2-\pi^4/64$ and $B=\pi^4/64.$
\end{rem}

\section{Neural networks}\label{sec:NN}

In this section we finally address the question  in which classes of neural networks   we can find the 
elements of the new Riesz basis $\mathcal{R}_d$.  Therefore, we first fix some  notation  and recall
what was shown in \cite[Sect.~6]{DDFHP} (in case $d=1$). \\

A function $f:\R^{n_1}\to \R^{n_2}$ is called affine, if it can be written as $f(x)=Mx+b$, where $M\in\R^{n_2\times n_1}$
is a matrix and $b\in\R^{n_2}.$ The following definition formalizes the notion
of $\ReLU$ neural networks with width $W$ and depth $L$, cf. Figure \ref{fig:NN} and \ref{fig:Neur}.
\begin{dfn}\label{dfn:NN} Let $d,W,L$ be positive integers. Then a feed-forward $\ReLU$ network ${\mathcal N}$ with width $W$
and depth $L$ is a collection of $L+1$ affine mappings $A^{(0)},\dots,A^{(L)}$, where $A^{(0)}:\R^d\to \R^W$,
$A^{(j)}:\R^W\to\R^W$ for $j=1,\dots,L-1$ and $A^{(L)}:\R^W\to \R$. Each such a network ${\mathcal N}$ generates
a function of $d$ variables
\[
A^{(L)}\circ \ReLU\circ A^{(L-1)}\circ\cdots\circ\ReLU\circ A^{(0)}.
\]
Moreover, we denote  by $\Upsilon^{W,L}$ the set of all functions, which are generated in this way
by some feed-forward $\ReLU$ network with width $W$ and depth $L$.
\end{dfn}

\begin{figure}[h!]
\begin{center}\includegraphics[width=12cm]{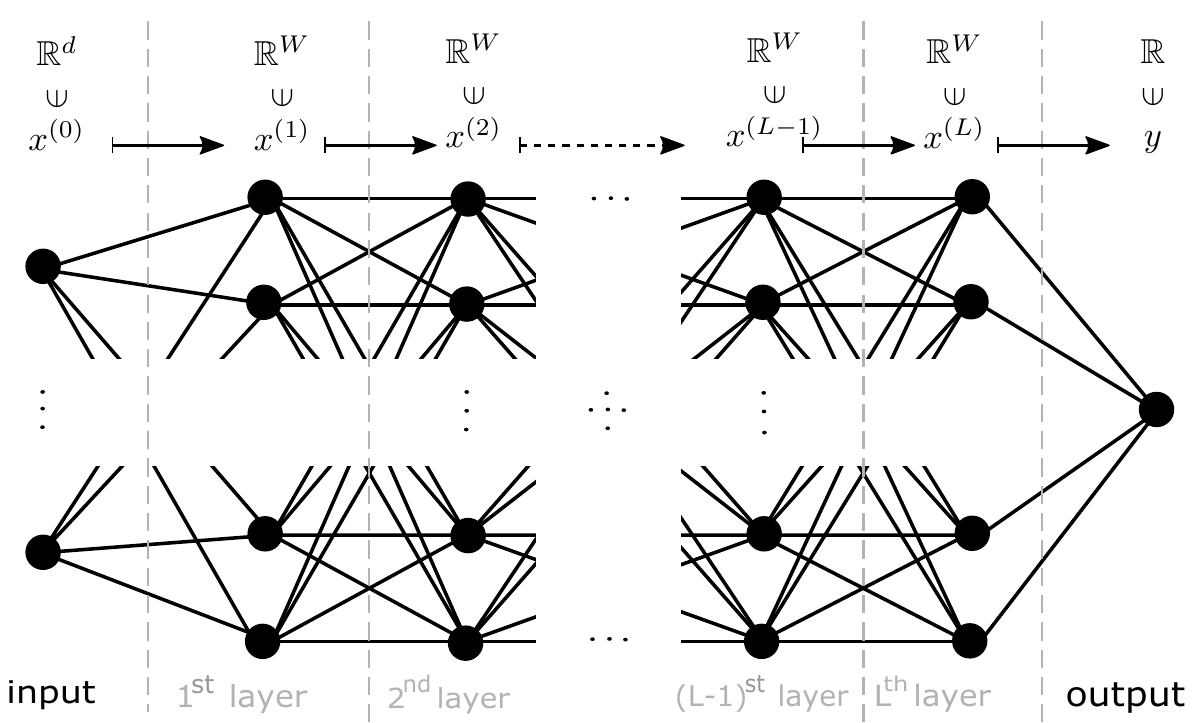}\\
\caption{Feed-forward $\mathrm{ReLU}$ network with length $L$, width $W$}
\label{fig:NN}
\end{center}
\end{figure}

\begin{figure}[h!]
\begin{center}\includegraphics[width=10cm]{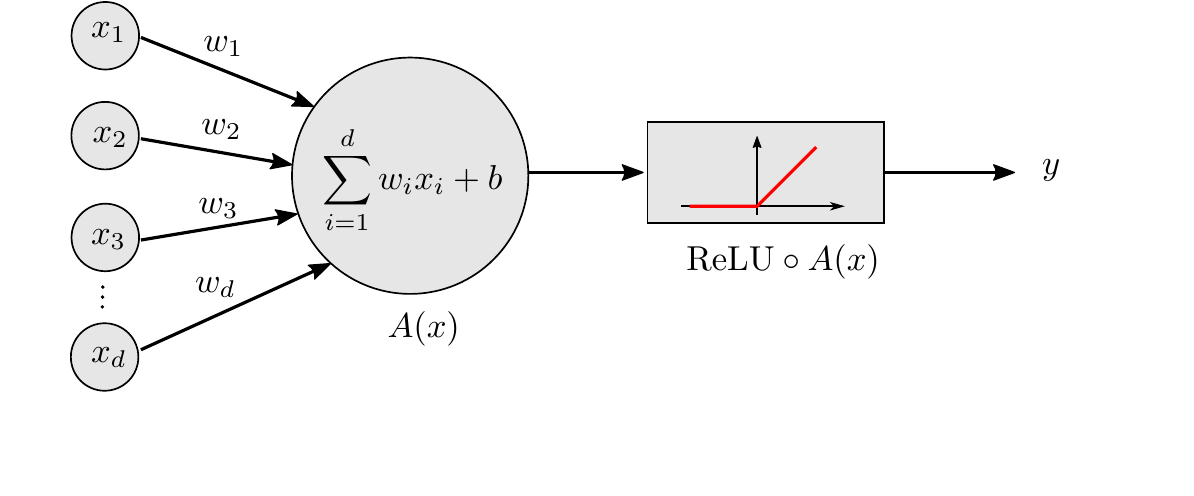}
\vskip-1cm
\caption{Close-up: $\mathrm{ReLU}$ activation function acting in neuron in 1st hidden layer}
\label{fig:Neur}
\end{center}
\end{figure}

Every $S\in \Upsilon^{W,L}$ is a continuous piecewise affine function on ${\mathbb R}^d$.
If the affine mappings associated to $S$ are denoted by $A^{(l)}$ with $l=0,\ldots, L$, then the value $S(x^{(0)})$
is computed for each input $x:=x^{(0)}\in \R^d$
after the calculation of a series of intermediate vectors $x^{(l)}:=\ReLU(A^{(l-1)}x^{(l-1)})\in \R^W$, $l=1,\dots,L$,
called vectors of activation at layer $l$. Finally the output $S(x)$ is produced as $S(x):=x^{(L+1)}=A^{(L)}x^{(L)}$. 

We collect some properties of the sets $\Upsilon^{W,L}$, which are needed in the sequel.
\begin{proposition}[]\label{prop_Y}
Let $W\geq 2$. 
\begin{itemize}
\item[(i)] Let $\mathcal{Y}_1\in \Upsilon^{W,L_1}, \ldots,  \mathcal{Y}_k\in \Upsilon^{W,L_k}$. 
Then the composition of the $\mathcal{Y}_i$ satisfies 
\[
\mathcal{Y}_k\circ \dots \circ \mathcal{Y}_1\in \Upsilon^{W,L}, \qquad L=L_1+\dots +L_k. 
\]
\item[(ii)] Let $L\ge 1$. Then $\Upsilon^{W,L}\subset\Upsilon^{W,L+1}$.
\end{itemize}
\end{proposition}
\begin{proof}
The proof of (i) can be found in \cite[Prop.~4.2]{DDFHP} for $d=1$.
The proof for general $d\ge 1$ follows virtually without any change.

For the proof of (ii), we consider the identity function ${\rm id}(x)=x$ for $x\in\R$. We rewrite it as
\begin{equation}\label{eq:idNN}
{\rm id}(x)=\begin{bmatrix}1 & -1\end{bmatrix}\ReLU\left\{\begin{bmatrix}1\\-1\end{bmatrix}x\right\}
\end{equation}
to conclude that ${\rm id}\in\Upsilon^{2,1}$. An easy modification of \eqref{eq:idNN} also shows that
${\rm id}\in\Upsilon^{W,1}$ for every $W\ge 2$. The result then follows by (i).
\end{proof}

An important example and building block for our constructions to follow is the \emph{hat function} $H:[0,1]\rightarrow \R$
\[
H(x):= \begin{cases}
2x,& 0\leq x\leq \frac 12,\\
2(1-x), & \frac 12<x\leq 1, 
\end{cases}
\]
which was already used in connection with
feed-forward neural networks with   $\ReLU$ activation function   by \cite{Telg},
cf. also \cite{DDFHP,Schm,Yarot}.
From the representation
\[
H(x)=\begin{bmatrix}2 & -4 \end{bmatrix}\mathrm{ReLU}\left\{
\begin{bmatrix} 
1\\1
\end{bmatrix}x+\begin{bmatrix}
0\\ -\frac 12
\end{bmatrix}
\right\},\quad x\in[0,1]
\]
we see, that $H$ belongs to $\Upsilon^{2,1}$, see Figure \ref{fig:hat-funct}.
\begin{figure}[h!]
\begin{center}
\includegraphics[width=12cm]{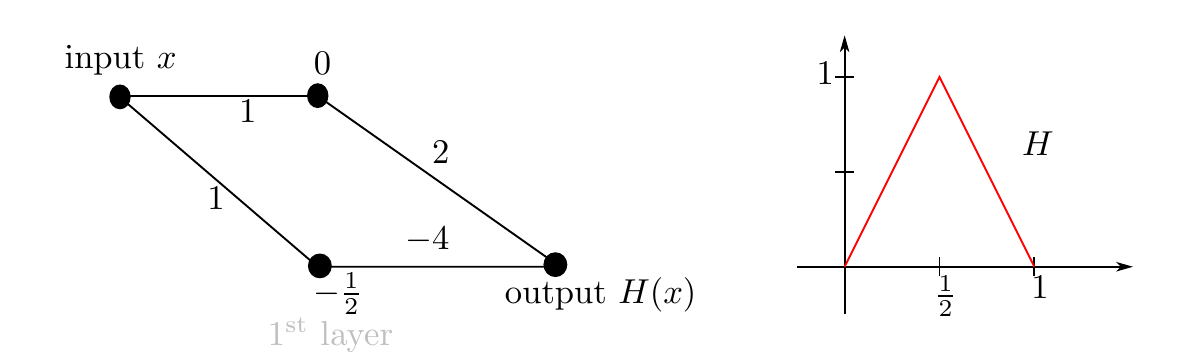}
\caption{Computational graph showing $H\in \Upsilon^{2,1}$ and usual graph associated with $H$}
\label{fig:hat-funct}
\end{center}
\end{figure}\\

Furthermore, since $H\in \Upsilon^{2,1}$, we deduce from  Proposition \ref{prop_Y}(i)    that the $k$-fold composition $H^{\circ k}:=H
\circ H\circ \dots \circ H$ belongs to $\Upsilon^{2,k}$,  cf. Figure \ref{fig:hat-funct2}. 

\begin{figure}[h!]
\begin{center}
\includegraphics[width=15cm]{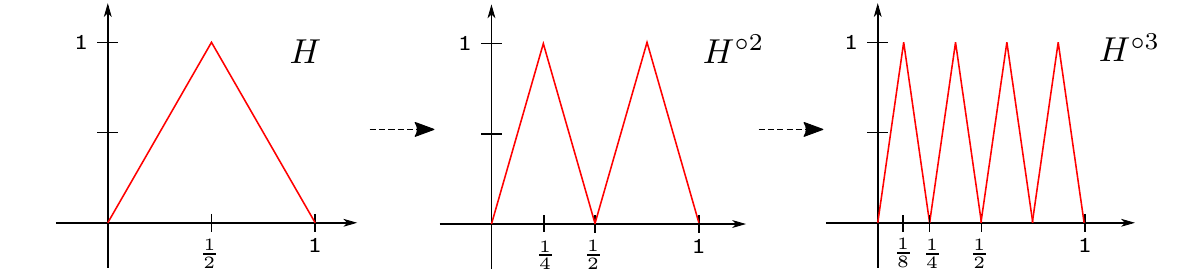}
\caption{The  graphs of  $H$, $H^{\circ 2}$, and $H^{\circ 3}$}
\label{fig:hat-funct2}
\end{center}
\end{figure}

In particular, $H^{\circ m}$ is a sawtooth function taking alternatively the values $0$ and $1$ at its breakpoints $l2^{-m}$,
$l=0,1,\ldots, 2^m$, cf. \cite[Lemma III.1]{EPGB} or \cite[Lemma 2.4]{Telg}. Moreover,  since  the restriction of the function
$(2^mx-\lfloor 2^m x\rfloor)$ on each interval $[l2^{-m}, (l+1)2^{-m})$ is a linear function passing through $l2^{-m}$ with
slope $2^m$ and $\CalC(0)=\CalC(1)$,  we have the coincidence 
\[
\CalC_{2^m}(x)=\CalC(2^mx-\lfloor 2^m x\rfloor)=\CalC(H^{\circ m}(x)), \quad x\in [0,1],
\]
 cf.  \cite[Page 147]{DDFHP}. 
Therefore, since $H$ and $\CalC=1-2H$ belong to $\Upsilon^{2,1}$,  we obtain from Proposition \ref{prop_Y}(i) that $\CalC_{2^m}\in \Upsilon^{2,m+1}$ for $m=\N_0$. Concerning  $\CalS$,   we deduce $\CalS\in \Upsilon^{2,2}$ from the identity $\CalS(x)=\CalC_2(\frac x2+\frac 38)$, $x\in [0,1]$,   which ultimately yields $\CalS_{2^m}\in \Upsilon^{2,m+2}$.
Moreover, for arbitrary $j\in \N$, we choose the smallest $m\in \N$ such that $j\leq 2^m$ and in view of $\CalC_j(x)=\CalC_{2^m}(j2^{-m}x)$, $j\leq 2^m$, see that the following holds.
\begin{lemma}\label{daub-lem-1}
Let $j\in \N$. Then,  restricted to $[0,1]$,
\[
\CalC_j\in \Upsilon^{2,\lceil \log_2 j\rceil+1}\quad \text{and}\quad 
\CalS_j\in \Upsilon^{2,\lceil \log_2 j\rceil+2}
\]
and all the entries of weight matrices and the bias vectors are bounded by 8.
\end{lemma}

\medskip 

This was already observed in the proof of \cite[Theorem~6.2]{DDFHP}.  We now provide a  multivariate version  of Lemma   \ref{daub-lem-1}.

\begin{lemma}\label{lem-mult-1}
Let $d>1$ and $\alpha\in \mathbb{Z}^d\setminus\{0\}$. Then,  restricted to $x\in[0,1]^d$, 
\[
\CalC(\alpha \cdot x)\in \Upsilon^{2,\lceil \log_2 \|\alpha\|_1\rceil+2}\quad \text{and}\quad 
\CalS(\alpha \cdot x)\in \Upsilon^{2,\lceil \log_2 \|\alpha\|_1\rceil+3},
\]
where $\|\alpha\|_1=|\alpha_1|+\ldots + |\alpha_d|$. 
Also in this case, the weights and biases are bounded by 8.
\end{lemma}

\begin{proof} We first extend the functions $\CalC_{2^m}$, $\CalS_{2^m}$ from $[0,1]$ to the interval $[-1,1]$ by putting
$$\tilde{\CalC}_{2^m}(x):=\CalC_{2^{m+1}}\left(\frac{x+1}{2}\right), \quad x\in [-1,1],$$ 
and deduce that $\tilde{\CalC}_{2^m}\in \Upsilon^{2,m+2}$. 
Let  now $x\in [0,1]^d$, then from 
$
\alpha \cdot x= \alpha_1x_1+\ldots + \alpha_d x_d 
$  we get 
\[
\alpha \cdot x\in \big[-\|\alpha\|_1, \|\alpha\|_1\big]. 
\]
Moreover, using the fact that  $\CalC(\alpha \cdot x)=\tilde{\CalC}_{2^m}(2^{-m}\alpha \cdot x)$, we choose $m:=\lceil \log_2\|\alpha\|_1\rceil$ and obtain 
\[
2^{-m}\alpha \cdot x\in \left[-\frac{\|\alpha\|_1}{2^m},\frac{\|\alpha\|_1}{2^m}\right]\subset[-1,1], 
\]
which implies $\CalC(\alpha \cdot x)\in \Upsilon^{2,\lceil \log_2\|\alpha\|_1\rceil+2}$.
The result for   $\CalS(\alpha \cdot x )$ follows by similar considerations. 
\end{proof}
\begin{rem}
Let us point out that we only have an implicit dependence of the length $L$ of our approximating 
neural network on the dimension $d$ (which is displayed by the fact that it logarithmically depends on the $\ell_1$-norm of $\alpha$).
\end{rem}

Finally, using Lemma \ref{lem-mult-1} we obtain a multivariate analogue of  \cite[Theorem~6.2]{DDFHP}, where it was shown  that one can reproduce linear combinations of $\CalC_k$ and $\CalS_k$ via $\mathrm{ReLU}$ networks with a good control of the depth $L$. 

\begin{thm} Let $d\ge 1$ and let $k,l\ge 0$ be integers with $k+l\ge 1$. Let $\{\alpha_1,\dots,\alpha_k,\beta_1,\dots,\beta_l\}\subset\Z^d\setminus\{0\}$.
Then the function
\[
f(x)=\sum_{i=1}^{k} a_i \CalC(\alpha_i\cdot x)+\sum_{j=1}^l b_j\CalS(\beta_j\cdot x),\quad x\in[0,1]^d
\]
belongs to $\Upsilon^{W,L}$ with
\[
W=2(k+l)\quad\text{and}\quad
L=\max_{\substack{i=1,\ldots, k;\\ j=1,\ldots, l\ }}\{\lceil\log_2(\|\alpha_i\|_1) \rceil+2, \lceil\log_2(\|\beta_j\|_1) \rceil+3\} 
\]
and the weights and biases in this network are bounded by 
$\max_{\substack{i=1,\ldots, k;\\  j=1,\ldots, l}}\{8|a_{i}|, 8|b_{j}|, 8\}$. 
\end{thm}
\begin{proof}
By Proposition \ref{prop_Y} (ii) and  Lemma \ref{lem-mult-1}, $\CalC(\alpha_i\cdot x)\in \Upsilon^{2,L}$ and
$\CalS(\beta_j\cdot x)\in\Upsilon^{2,L}$ for every $i=1,\dots,k$ and $j=1,\dots,l$ if we restrict $x$ to $[0,1]^d$.
By Definition \ref{dfn:NN} we have the corresponding representations for $x\in[0,1]^d$ and all admissible $i$'s and $j$'s
\begin{align}
\label{eq:NN_p1} \CalC(\alpha_i\cdot x)&=A_i^{(L)}\circ \ReLU\circ A_i^{(L-1)}\circ\cdots\circ\ReLU\circ A_i^{(0)}(x)\\[-.5cm]
\intertext{and}\notag\\[-1.1cm]
\label{eq:NN_p2}\CalS(\beta_j\cdot x)&=B_j^{(L)}\circ \ReLU\circ B_j^{(L-1)}\circ\cdots\circ\ReLU\circ B_j^{(0)}(x).
\end{align}
If $z=(z_1,\dots,z_n)$ is a vector in $\R^n$ and $1\le u<v\le n$ are integers, then we denote by
$z_{u,v}=(z_u,z_v)$ the restriction of $z$ to the set $\{u,v\}$. Furthermore, we denote
$x^{(0)}=(x_1,\dots,x_d)$ and stack the networks \eqref{eq:NN_p1} and \eqref{eq:NN_p2} on top of each other. In this way, we
obtain a series of intermediate vectors $x^{(1)},\dots,x^{(L)}\in\R^{W}$
\begin{align*}
x^{(1)}&=\ReLU\begin{bmatrix} A_1^{(0)}x^{(0)}&\dots&A^{(0)}_kx^{(0)}&B_1^{(0)}x^{(0)}&\dots&B^{(0)}_lx^{(0)}\end{bmatrix}^T,\\
x^{(2)}&=\ReLU\begin{bmatrix} A_1^{(1)}x_{1,2}^{(1)}&\dots&A^{(1)}_kx_{2k-1,2k}^{(1)}
&B_1^{(1)}x_{2k+1,2k+2}^{(1)}&\dots&B^{(1)}_lx_{2k+2l-1,2k+2l}^{(1)}\end{bmatrix}^T,\\
&\vdots\\
x^{(L)}&=\ReLU\begin{bmatrix} A_1^{(L-1)}x_{1,2}^{(L-1)}&\dots&A^{(L-1)}_kx_{2k-1,2k}^{(L-1)}
&B_1^{(L-1)}x_{2k+1,2k+2}^{(L-1)}&\dots&B^{(L-1)}_lx_{2k+2l-1,2k+2l}^{(L-1)}\end{bmatrix}^T.
\end{align*}
Finally, the result follows by observing that
\[
y=f(x)=\sum_{i=1}^k a_i A_i^{(L)}x^{(L)}_{2i-1,2i}+\sum_{j=1}^l b_j B_j^{(L)}x^{(L)}_{2j-1+2k,2j+2k},
\]
see Figure \ref{fig:thm-4-5}.\vskip-.7cm
\begin{figure}[h!]
\begin{center}
\includegraphics[width=15cm]{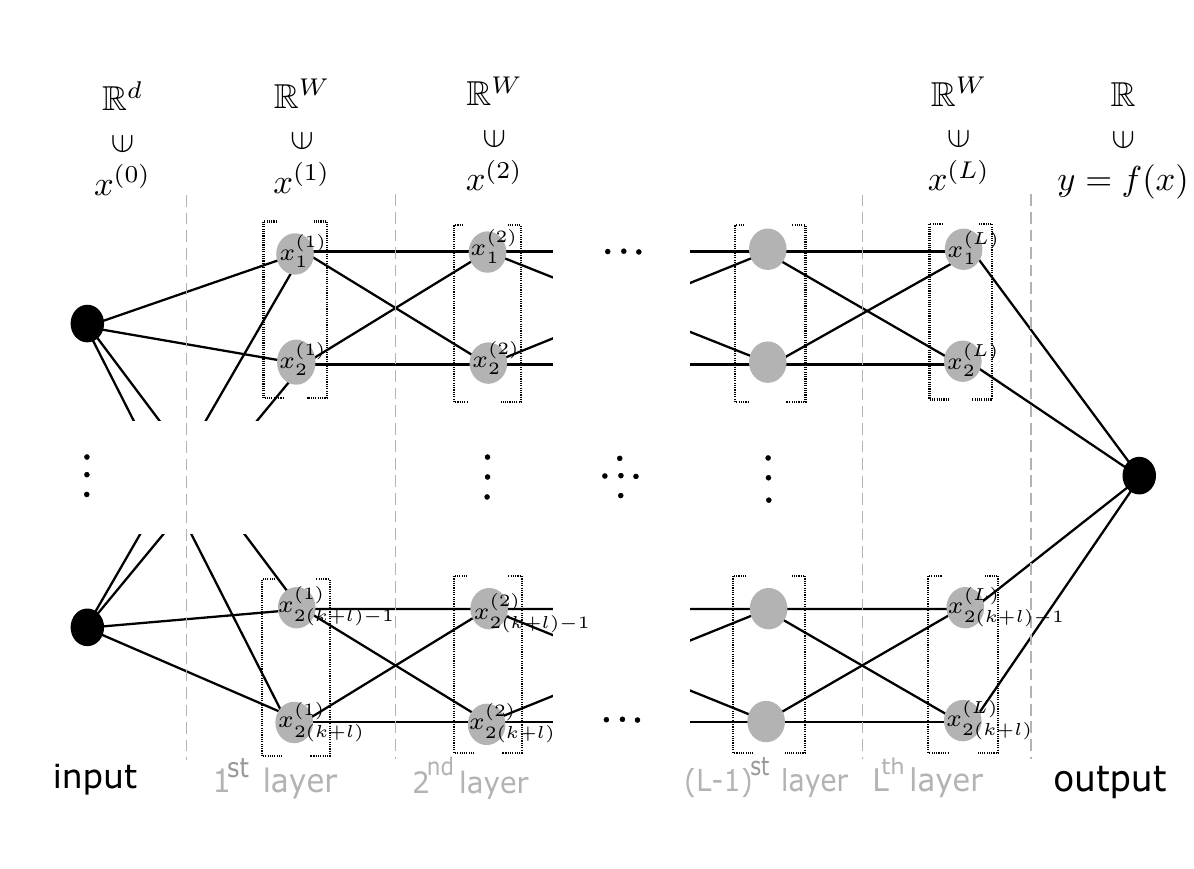}\vskip-0.8cm
\caption{Feed-forward ReLU network producing $y=f(x)$}
\label{fig:thm-4-5}
\end{center}
\end{figure}
Since by Lemma \ref{lem-mult-1}, all the weights and biases used in the calculation of $x^{(1)},\dots,x^{(L)}$
are bounded by 8, the weights in the last step are then bounded by $\max_{i=1,\dots,k}8|a_i|$ and
$\max_{j=1,\dots,l}8|b_j|$, respectively.
\end{proof}

{\bf Acknowledgment:} We would like to thank the authors of \cite{DDFHP} for their kind permission to re-use some of their
figures. We also thank Dorothee Haroske (FSU Jena, Germany) for her hospitality during our stay in Jena, where part of the work
took place.

\end{document}